\documentclass[letterpaper,10pt,journal,final,twocolumn]{IEEEtran}  
 
\usepackage{amssymb,amsmath,color,enumitem}
\usepackage[english]{babel}
\usepackage{graphicx,epstopdf,float,tikz}
\usetikzlibrary{arrows,positioning,calc}
 \usepackage[noadjust]{cite}
 \usepackage{makecell}
 \usepackage{multirow}

 \newcommand{\QED}{\blacksquare}

\newtheorem{assumptionx}{Assumption}
\newtheorem{lemmax}{Lemma}
\newtheorem{propositionx}{Proposition}
\newtheorem{remarkx}{Remark}
\newtheorem{theoremx}{Theorem} 
\newtheorem{definitionx}{Definition} 
\newtheorem{examplex}{Example}

\newenvironment{definition}{\begin{definitionx}}{\ignorespaces\null\hfill$\triangleleft$\end{definitionx}}
\newenvironment{example}{\begin{examplex}}{\ignorespaces\null\hfill$\triangleleft$\end{examplex}}
\newenvironment{theorem}{\begin{theoremx}\itshape }{\ignorespaces\null\hfill$\triangleleft$\end{theoremx}}
\newenvironment{remark}{\begin{remarkx}}{\ignorespaces\null\hfill$\triangleleft$\end{remarkx}}
\newenvironment{proposition}{\begin{propositionx}\itshape }{\ignorespaces\null\hfill$\triangleleft$\end{propositionx}}
\newenvironment{lemma}{\begin{lemmax}\itshape }{\ignorespaces\null\hfill$\triangleleft$\end{lemmax}}
\newenvironment{assumption}{\begin{assumptionx}\itshape}{\end{assumptionx}}
\newenvironment{proof}[1][Proof]{\par\emph{#1. }}{\ignorespaces\hfill$\QED$}
\newenvironment{skproof}{\par\emph{Sketch of Proof. }}{\ignorespaces\hfill$\QED$}
\newcounter{assc}
\newcommand{\R}{\mathbb R}
\newcommand{\B}{\mathbb B}

\newcommand{\K}{\mathbb K}

\newcommand{\N}{\mathbb N}

\newcommand{\Rplus}{\R_{\ge 0}}

\newcommand{\KL}{\mathbb{KL}}
 
\newcommand{\cA}{\mathcal A}

\newcommand{\cC}{\mathcal C}
\newcommand{\cD}{\mathcal D}

\newcommand{\cF}{\mathcal F}

\newcommand{\cK}{\mathcal K}
\newcommand{\cL}{\mathcal L}

\newcommand{\cN}{\mathcal N}
\newcommand{\cO}{\mathcal O}
\newcommand{\cP}{\mathcal P}
\newcommand{\cQ}{\mathcal Q}
\newcommand{\cR}{\mathcal R}
\newcommand{\cS}{\mathcal S}
\newcommand{\cT}{\mathcal T}

\newcommand{\cW}{\mathcal W}
\newcommand{\cX}{\mathcal X}

\newcommand{\x}{\times}
\newcommand{\st}{\,\mid\,}

\newcommand{\und}{\underline}
\newcommand{\inv}{^{-1}}

\DeclareMathOperator{\rank}{rank}

\DeclareMathOperator{\col}{col}

\DeclareMathOperator{\dom}{dom}

\DeclareMathOperator{\range}{ran}

\newcommand{\overbar}[1]{\mkern 1.5mu\overline{\mkern-1.5mu#1\mkern-0.0mu}\mkern 1.5mu}

\newcommand{\interior}[1]{\mathring{#1}}
\newcommand{\closure}[1]{\overbar{#1}}
\newcommand{\ball}{\B}

\newcommand{\xb}{\mathbf{x}}
\newcommand{\sr}{^\star}

\newcommand{\y}{y}

\newcommand{\vhi}{\varphi}

\newcommand{\xp}{x_p}
\newcommand{\dotxp}{\dot x_p}
\newcommand{\yp}{y}
\newcommand{\fp}{f_p}
\newcommand{\hp}{h_p}
\newcommand{\dxp}{{n_p}}

\newcommand{\xg}{\xi}
\newcommand{\yg}{y_\xi}
\newcommand{\dotxg}{\dot \xi}
\newcommand{\fg}{f_\xi}
\newcommand{\hg}{h_\xi}
\newcommand{\dxg}{{n_\xi}}
\newcommand{\Xg}{\Xi}

\newcommand{\xw}{x}

\newcommand{\dotxw}{\dot x}
\newcommand{\fw}{f}
\newcommand{\hw}{h}
\newcommand{\dxw}{{n_x}}
\newcommand{\Xw}{X}
\newcommand{\Fw}{F}

\newcommand{\cim}{\eta_{\rm im}}
\newcommand{\dotcim}{\dot{\eta}_{\rm im}} 
\newcommand{\Acim}{\Phi}
\newcommand{\Bcim}{G} 
\newcommand{\cst}{\eta_{\rm st}}
\newcommand{\dotcst}{\dot{\eta}_{\rm st}}  
\newcommand{\fcst}{\phi_{\rm st}}

\newcommand{\Xb}{\mathbf{X}}
\newcommand{\Ub}{\mathbf{U}}
\newcommand{\pb}{\mathbf{p}}
\newcommand{\Pb}{\mathbf{P}}

\newcommand{\signal}{\psi}

\newcommand{\yaux}{y_{\rm a}}

\newcommand{\nominal}[1]{#1^{\circ}}

\newcommand{\dx}{{n_x}} 
\newcommand{\dy}{{n_{\y}}}
\newcommand{\du}{{n_u}}

\newcommand{\dc}{{n_c}}

\newcommand{\dw}{{n_w}}
\newcommand{\de}{{n_e}} 
\newcommand{\dxc}{{n_c}}
\newcommand{\dcim}{{n_{\rm im}}}
\newcommand{\dcst}{{n_{\rm st}}}
\newcommand{\dyaux}{{n_{\rm a}}}

\newcommand{\vep}{\varepsilon}

\newcommand{\SST}{\cO}

 \renewcommand{\P}{\cP}

 \newcommand{\propdef}[1]{``\ #1\ "}

\allowdisplaybreaks

\newcommand{\MB}[1]{{#1}}

\newcommand{\fourier}{{\rm c}}

\newenvironment{cellitemize}{\begin{enumerate}[wide=0pt, widest=999,leftmargin=\parindent,labelsep=*,leftmargin=15pt]}{\end{enumerate}\vspace*{-\baselineskip}\leavevmode}

\title{
About Robustness of Control Systems\\ Embedding an Internal Model
}
\author{Michelangelo~Bin,~Daniele Astolfi~and~Lorenzo Marconi%
	\thanks{%
		Michelangelo Bin (m.bin@imperial.ac.uk) is with the Department of Electrical and Electronic Engineering, Imperial College London, UK; Daniele Astolfi (daniele.astolfi@univ-lyon1.fr) is with Univ Lyon, Universit\`{e} Claude Bernard Lyon 1, CNRS, LAGEPP UMR 5007; Lorenzo Marconi (lorenzo.marconi@unibo.it) is with the CASY-DEI, University of Bologna, Italy.  
	}
 
}

\begin{document}

\onecolumn 
\vspace{4em}

\vspace{5em}
\begin{quote}
	\emph{\textcopyright{}~2021 IEEE.  Personal use of this material is permitted.  Permission from IEEE must be obtained for all other uses, in any current or future media, including reprinting/republishing this material for advertising or promotional purposes, creating new collective works, for resale or redistribution to servers or lists, or reuse of any copyrighted component of this work in other works.}
\end{quote}

\twocolumn
 
\maketitle

\begin{abstract}
Robustness is a basic property of any control system. In the context of linear output regulation, it was proved that  embedding an internal model of the exogenous signals is necessary and sufficient to achieve tracking of the desired reference signals in spite of external disturbances and parametric uncertainties. This result is commonly known as ``internal model principle''. A complete extension of such linear result to  general nonlinear systems is still an open problem, which is exacerbated by the  large number of alternative  definitions of uncertainty and desired control goals that are possible in a nonlinear setting. In this paper, we develop a   general framework in which all these different notions can be formally characterized in a unifying way. Classical results are reinterpreted in the proposed setting, and new results and insights are presented with a focus on robust rejection/tracking of arbitrary harmonic content. Moreover, we show by counter-example that, in the relevant case of continuous unstructured uncertainties, there are problems for which no smooth finite-dimensional robust regulator exists ensuring exact regulation. 
\end{abstract}

\section{Introduction} 
\subsection{Problem Overview}
We consider nonlinear systems of the form 
\begin{equation}\label{s.wxy}
\begin{array}{lcl}
\dot w &=& s(w)\\
\dotxp&=& \fp(w,\xp,u)\\
\yp &=& \hp(w,\xp),
\end{array}
\end{equation} 
in which $\xp\in\R^\dxp$ is the state of the \emph{plant}, $\yp\in\R^\dy$ is the \emph{measured output},  $u\in\R^\du$ is the \emph{control input}, and $w\in\R^\dw$ models a set of unmeasured \emph{exogenous inputs}   representing  disturbances, uncertain parameters, reference commands to be tracked, or any other signal coming from the ``external world''. 
For system \eqref{s.wxy} we assume that, on the basis of the nominal value  of the   maps $(s,\fp,\hp)$  and of the set  in which $(w,\xp)$ is supposed to originate, a \emph{finite-dimensional} output feedback regulator of the form
\begin{equation}\label{s.xc}
\begin{array}{lclclcl}
\dot x_c &=& f_c(x_c,\yp), && u&=&h_c(x_c,\yp)
\end{array}
\end{equation}
has been designed to guarantee that the ``nominal'' closed-loop trajectories asymptotically fulfill  a given desired property~$\P$. As better detailed later, the property $\P$ represents the desired asymptotic behavior for the closed-loop system \eqref{s.wxy}-\eqref{s.xc}, and it may coincide, for instance,  with a state of equilibrium, with the tracking of some reference signals or, more in general, with a desired   steady-state optimality condition. The general problem that we consider  is   to study the performance of the closed-loop system, in terms of the achievement of the property~$\P$,    when the controller \eqref{s.xc} is applied to a \emph{perturbed plant}, whose   maps and initial conditions do not necessarily match with their nominal values on which the regulator is tuned. Our aim is thus to characterize and design controllers  guaranteeing that, for ``small enough'' perturbations, the closed-loop trajectories still fulfill the asymptotic property $\P$, as in the nominal case.

\MB{For mathematical reasons, essentially related to the properties of limit sets and of the class of solutions considered, we restrict our focus to continuous systems,  thus leaving out discontinuous plants or controllers. We remark, however, that the   introduced framework can be extended to cover discontinuous systems by considering hybrid inclusions (see, e.g., \cite{Goebel2012book}) instead of limiting to differential equations. As this extension does not change the main message of the paper, and it is instead associated with a considerable technical overhead, we leave it for future research.}

Inspired by the \emph{output regulation} literature, in the following we  often link the property $\cP$ to the asymptotic value of a \emph{regulation error} $e$, defined as 
\begin{equation}\label{d.e}
e = h_e(w,\xp) \,\in\R^\de,
\end{equation}
which may or may not be measured (i.e. be a part of $y$). In this case, the fulfillment of Property $\cP$ pertains the design of a regulator ensuring (i) boundedness of the closed-loop trajectories, and (ii) that a prescribed \emph{asymptotic bound} holds on the regulation error. Depending on the nature of this bound, the following taxonomy is used.  \emph{Approximate regulation} requires 
	\begin{equation}\label{b.ls_e_vep}
	\limsup_{t\to\infty}|e(t)|\le \varepsilon,
	\end{equation}
	in which $\varepsilon\ge 0$ is  a possibly small number measuring the  performances of the controller (in this case $\vep$ is allowed to depend on the particular solution). \emph{Asymptotic regulation} refers to the relevant  case in which   \eqref{b.ls_e_vep} holds with $\vep=0$. \emph{Practical regulation}, instead, is achieved when $\varepsilon$ can be taken arbitrarily in $(0,\infty)$ by tuning, accordingly, some control  parameters. 

Traditionally, in  output regulation, the concept of robustness almost exclusively refers to asymptotic regulation. In particular, a \emph{robust regulator} is a regulator that ensures $\lim_{t\to\infty} e(t)=0$ even under ``small enough" perturbations of the plant's maps $\fp$, $\hp$, and $h_e$, with the exosystem map $s$ that, however, remains untouched (see \cite{Francis1976,Davison1976,Huang1994,Byrnes1997,jie_huang_remarks_2001,Astolfi2017} and the references therein).
This notion of robustness was coined in the 70s in the context of linear systems, under the name of ``structurally stable regulation'' \cite{Francis1976,Davison1976}. The importance of keeping $s$ nominal is motivated by the fact that, in a linear setting, knowing $s$ is the main  necessary~\cite{Francis1976} and sufficient~\cite{Davison1976} condition for the  design of a robust  regulator (in particular, necessity is known as the \emph{internal model principle}~\cite{Francis1975,Francis1976}).
When perturbations affect the  map $s$ of the exosystem, instead, the problem is typically referred to as \emph{adaptive} regulation  \cite{Serrani2001,marino_output_2003,DelliPriscoli2006,forte_robust_2017,Bin2019adaptlinear,Bin_classtype_2019,bin_approximate_2020,bernard_adaptive_2020}. 
For the special role played by $s$ in the design of robust linear regulators,  this distinction between robust and adaptive regulation makes considerable sense in the context of linear output regulation. However, it  completely loses meaning in a general nonlinear setting, in which the role of the exosystem and that of the plant mix up in determining the {steady-state} signals that the regulator must reproduce to keep the regulation error to zero \cite{Byrnes2003,bin_chicken-egg_2018,Bin2019}.

Overall, the problem of designing robust  regulators achieving the property of asymptotic regulation for general \emph{nonlinear} systems is extremely more difficult than the linear case (if not impossible as indicated later in Section~\ref{sec.nonlinear}) and, at the best of  our knowledge,  \emph{no} smooth finite-dimensional robust asymptotic regulator currently exists in the literature for nonlinear systems. This inherent difficulty has recently motivated a  shift towards  new approaches seeking approximate, rather than asymptotic,   results, as they trade   weaker  claims  for  stronger robustness properties (see, e.g. \cite{Isidori2012,Astolfi2015,bin_chicken-egg_2018,Bin2019,Bin_classtype_2019,bin_approximate_2020,bernard_adaptive_2020}). 
Nevertheless, such robustness properties are usually treated in ad hoc manners, poorly characterized, and sometimes  not formally proved and just left to intuition. In turn, within the current output regulation community, a rigorous unifying characterization of robustness  including all these new approximate regulation properties   is definitely missing.

Outside output regulation, many notions and results related to robustness exist that cover  a number of cases of interest. 
For instance, in the branch of mathematics studying dynamical systems robustness is usually framed in the context of \emph{structural stability}  (see e.g. \cite{Smale1967,Zeeman1988}), in which a system $\dotxg=f(\xg)$, with $f$ smooth, is said to be \emph{structurally stable} if there exists a neighborhood $\cN$ of $f$ in the $\cC^1$ topology \cite{Hirsch1994} such that, for each $\tilde f\in\cN$, the orbits of the solutions to $\dotxg=\tilde{f}(\xg)$ are homeomorphic to those of the ``nominal'' system $\dotxg=f(\xg)$ through a homeomorphism  $\phi$  that preserves orientation and that can be made as ``close'' to the identity as desired by opportunely restricting~$\cN$.  
Another important notion of robustness is the concept of \emph{robustness of asymptotic stability} \cite[Chapter 7]{Goebel2012book}, in which a set  that is asymptotically stable for the nominal system, remains asymptotically stable also if the plant is slightly perturbed. Robust asymptotic stability is usually characterized by Lyapunov methods, which are also intensively used  in the synthesis of robust controllers \cite[Chapter 7]{Goebel2012book}. For systems with inputs, robust asymptotic stability of a set is generalized by the concept of ``robust input-to-state stability''~\cite{Cai2013}. 
Other notions of robustness,  more focused on  synthesis, may be found  in the context of ``robust control'', and in particular in the fields of  ``$H_\infty$ control''~\cite{Zhou1996}, ``high-gain'' stabilization \cite{Teel1995}, or \MB{``variable structure'' and ``sliding mode'' methods~\cite{shtessel_sliding_2014}.}
All these frameworks, however, fail to capture, in its full generality, the entire set of asymptotic behaviors and properties of interest in output regulation.

\subsection{A Motivating Example}
Consider a special case of \eqref{s.wxy} in which 
\begin{equation}\label{s.ex.wxy}
\begin{array}{lcl}
\dot w &=& \varrho \begin{pmatrix}
0 & 1\\-1 & 0
\end{pmatrix}w\\
\dotxp &=& q(w) + \alpha(\xp) +u\\
y &=& \xp 
\end{array}
\end{equation}
for some $\varrho>0$ and some continuous functions $q$ and $\alpha$,  with $\alpha(0)=0$. The plant  is subject to a disturbance $q(w(t))$ generated by a nonlinear function of a single harmonics at frequency $\varrho$.
Assume first that the desired control goal is to stabilize the plant's state $\xp$ to the origin (i.e., in the language of this paper, to guarantee asymptotically a property $\P$ of the kind $\propdef{\xp=0}$). Then, we may think to every regulator of the form \eqref{s.xc} solving the problem as providing a control action $u(t)$ composed of two parts: \textbf{(i)} a \emph{stabilizing} action counteracting the potentially destabilizing term $\alpha(\xp)$; \textbf{(ii)} a compensation action asymptotically balancing the disturbance $q(w(t))$.

Stabilization of the origin of $\xp$ is a canonical   asymptotic regulation problem for $e:=\xp$, and it falls in the   setting  described above.
In turn,  if   $\varrho$ is known, $q$ and $\alpha$ are linear, and every admissible perturbation   does not destroy their linearity, then a robust solution is directly provided by means of internal model arguments, see e.g.~\cite{Davison1976,Byrnes1997}. If $\varrho$ is unknown, but linearity of $q$ and $\alpha$ is preserved, then some considerable degree of robustness is obtained by employing adaptive or immersion techniques (see for instance~\cite{Serrani2001,Isidori2012,Bin2019adaptlinear}).
On the other hand, if  perturbations of $q$ destroying its linearity are admitted, then no  regulator is available in the  literature that guarantees $\xp\to 0$ robustly. Moreover, later in Section~\ref{sec.nonlinear} we also show that, in the relevant case in which perturbations of $q$ are intended in the $\cC^0$ topology \cite{Hirsch1994}, then this problem \emph{admits no robust solution} even if $\varrho$ is known.

This negative result, in turn,  motivates us to   look  for  asymptotic properties that are weaker than $\xp=0$, but that  make the robustness problem well posed and solvable. For instance, instead of yearning for $\xp=0$,  a  weaker objective of robust rejection in $\xp(t)$ of the dominant harmonic at frequency $\varrho$ could be asked (this new objective clearly coincides with the former  in case $q$ and $\alpha$ are linear). The robust regulation problem associated with this new asymptotic property   thus pertains the design of a controller~\eqref{s.xc} ensuring that the amplitude of the harmonic at frequency $\varrho$ of the signal $\xp(t)$ vanishes asymptotically even under sufficiently small ``nonlinear perturbations'' of $q$.

This weaker ``robust harmonic rejection'' objective  however falls outside the canonical scope of robust output regulation theory. Besides, it is   also not clear how  it might be characterized in terms of structural stability or  robust asymptotic stability of a pre-specified nominal set. In fact,   there exists  no general regulation framework  in which this kind of robustness notions can be   formally described.

\subsection{Contributions and Organization of the Paper} 
Motivated by the above discussion, we develop a new  framework  in which the various robustness properties of actual approximate and asymptotic design solutions can be formally described and characterized in a unifying nonlinear setting. General classes of ``unstructured'' perturbations,  both of the system's dynamics and in the set of initial conditions,  are formally defined in the context of {\em general topology}. A notion of ``steady-state property'' is then presented, along with  conditions  under which such a property is robustly preserved. 
 Classical definitions of robustness and asymptotic/approximate output regulation results are   reinterpreted within the proposed setting, and new results and insights concerning the robust rejection  of arbitrary harmonics are given. 
 Finally, a counter-example is presented showing that there exist problems for which \emph{no smooth finite-dimensional robust regulator exists that can guarantee robust asymptotic regulation}, if perturbations are meant in a canonical ``$\cC^0$ sense''. This latter result shows that, in a general nonlinear context, asymptotic regulation is a property that cannot be obtained robustly by regulators of finite dimension, and that approximate regulation, rather than asymptotic, is the correct way of approaching the problem in a nonlinear uncertain setting if robustness is sought.  
 
The article is organized as follows. Section~\ref{sec.Prelimiraies} contains the notation and some preliminary definitions. In Section~\ref{sec.Framework} we build the basic robustness framework, by formally defining the concepts of steady-state properties, perturbations and robustness. In Section~\ref{sec.Linear}, we present some results concerning regulation schemes embedding a linear internal model. Finally, in Section~\ref{sec.nonlinear}, we consider regulation schemes embedding nonlinear internal models, and we present the aforementioned counterexample.
   
\section{Preliminaries}\label{sec.Prelimiraies}
\subsection{Notation and Basic Definitions}
We denote by $\R$ and $\N$ the set of real and natural numbers respectively, and we let $\Rplus:=[0,\infty)$. The symbol $|\cdot|$ is used to denote norms when the specific normed space is clear from the context. For $x\in\R^n$ and $A\subset\R^n$, $n\in\N$, $|x|_A:=\inf_{a\in A}|x-a|$ denotes the distance of $x$ to $A$. If $f$ is a function, we denote by $\dom f$ its domain and by $\range f$ its range. If $A$ is a linear map,   $\sigma(A)$ denotes its spectrum. The symbol $\subset$ denotes non-strict inclusion. 
If $A,\,B\subset\R^n$, we let $A+B:=\{a+b\st a\in A,\, b\in B\}$. With $\epsilon\in\R$, we let $\epsilon A:=\{\epsilon a\st a\in A\}$. When the underlying metric space is clear, we denote by $\ball$ the open ball of radius $1$. 
 
For basic concepts  about topological spaces we refer to \cite{Engelking1989}. In particular, a \textit{topology} $\tau$ on a set $S$ is a family of subsets of $S$ satisfying: \textit{(i)} $\emptyset,\,S\in\tau$; \textit{(ii)} $\tau$ is closed under arbitrary unions; \textit{(iii)} $\tau$ is closed under finite intersections.
The elements of $\tau$ are called {\it open sets}. A \emph{topological space} is a pair $(S,\tau)$ in which $S$ is a set and~$\tau$ a topology on $S$. Given a topological space $(S,\tau)$ and an element $s\in S$, a subset $U\subset S$ is called a {\it $\tau$-neighborhood}  of $s$ if it contains an open set containing $s$.  If $A\subset S$, a set $U\subset S$ is called a neighborhood of $A$ if it is a neighborhood of each point of $A$. When $\tau$ is clear from the context, we simply write $S$ in place of $(S,\tau)$. Given $A\subset S$, the collection $\tau_A=\{ U\cap A\st U\in\tau\}$ is called the \emph{subspace topology} induced by $\tau$, and the pair $(A,\tau_A)$  is said to be a \emph{topological subspace} of $(S,\tau)$.
With $A\subset S$, we denote by  $\interior{A}$ its interior, by $\closure{A}$ its closure, and by $\cK(A)$ the set of all its compact subsets. A \emph{directed set} is a pair $(A,\preceq)$ in which $A$ is a set and~$\preceq$ a \emph{preorder} on $A$. When $\preceq$ is clear, we omit it and we denote $(A,\preceq)$ simply by $A$. A \emph{net} on $S$ is a function from a directed set $A$ to $S$. A sequence is a net with $(A,\preceq)=(\N,\le)$. We denote a net $s:A\to S$ equivalently by $\{s_\alpha\}_{\alpha \in A}$ or just by $\{s_\alpha\}_\alpha$ or $\{s_\alpha\}$ when $A$ is clear. When $S$ is given a topology~$\tau$, then $\{s_\alpha\}$ is said to converge to a point $\bar s\in S$ if for every $\tau$-neighborhood $U$ of $\bar s$ there exists $\bar\alpha_U$ such that $s_\alpha\in U$ for all $\alpha\in A$ satisfying $\bar\alpha_U\preceq \alpha$. 
 With $k\in\N$ and $A$ and $B$ metric spaces, we denote by $\cC^k(A,B)$ the set of $k$-times continuously differentiable functions from $A$ to $B$. In particular, $\cC^0(A,B)$   denotes the set of continuous functions $A\to B$. When $A$ and $B$ are clear from the context we will omit the arguments and write $\cC^k$. If $x\in\cC^k(\R,\R^n)$, $n\in\N$, $\dot x:=dx/dt$, $x^{(0)}=x$, and, for $i=1,\dots,k$, $x^{(i)}:=d^i x/dt^i$. A function $\rho\in\cC^0(\Rplus,\Rplus)$ is said to be of class-K ($\rho\in\K$) if $\rho(0)=0$ and it is strictly increasing. A function $\beta\in\cC^0(\Rplus\x\Rplus,\Rplus)$ is said to be of class-KL ($\beta\in\KL$) if $\beta(\cdot,t)\in\K$ for all $t\ge 0$, and $\beta(s,\cdot)$ is strictly decreasing to zero for all $s\ge 0$.

\subsection{Systems and Limit Sets}

Consider a system of the form 
\begin{equation}\label{s.xi_preliminary}
\Sigma:\;\dotxg=\fg(\xg),
\end{equation}
defined on $\R^\dxg$, with $\dxg\in\N$. A solution to \eqref{s.xi_preliminary} is called \emph{maximal} if it cannot be continued further, and \emph{complete} if its   domain is unbounded. Given a subset $\Xg\subset\R^\dxg$, we denote by $\cS_\Sigma(\Xg)$ the set of all the maximal solutions to \eqref{s.xi_preliminary} originating in $\Xg$. When $\Xg=\R^\dxg$ we omit the argument and we write $\cS_\Sigma$. By convention we set $\emptyset=\cS(\emptyset)$. We define the \emph{reachable tails of $\Sigma$ from $\Xg$} as the  sets
\begin{equation*} 
\cR_\Sigma^t(\Xg):=\big\{ \bar\xg\in\R^\dxg \st \bar\xg=\xg(s),\,\xg\in\cS_\Sigma(\Xg),\,s\ge t\big\}
\end{equation*}
obtained for $t\ge 0$.  The net $t\mapsto  \cR_\Sigma^t(\Xg)$ is decreasing, in the sense that $t_1\ge t_2$ implies $\cR_\Sigma^{t_1}(\Xg)\subset \cR_\Sigma^{t_2}(\Xg)$. Therefore, the following quantity is well defined  (although possibly empty)
\begin{equation*} 
\Omega_\Sigma(\Xg) := \bigcap_{t\ge 0} \closure{\cR_\Sigma^t(\Xg)}.
\end{equation*}
The set $\Omega_\Sigma(\Xg)$ is called  the \emph{limit set} of $\Sigma$ from $\Xg$.

We say that $\Sigma$ is {\it uniformly ultimately bounded} from $\Xg$ if there exist a bounded subset $K\subset\R^\dxg$ and a $t\ge 0$ such that $\cR_\Sigma^t(\Xg)\subset K$ and $\cR_\Sigma^t(\Xg)\ne \emptyset$.  We say that a set $A$   \emph{uniformly attracts $\Sigma$ from $\Xg$}  if, for each neighborhood $U$ of $A$, there exists $t\ge 0$ such that $\cR_\Sigma^t(\Xg)\subset U$. 
We say that $A$ is \emph{forward invariant for $\Sigma$} if $\cR_\Sigma^0(A)\subset A$.
The limit set $\Omega_\Sigma(\Xg)$ has the following   properties \cite[Proposition 6.25]{Goebel2012book}.
\begin{proposition}\label{prop.Omega}Let $\Xg$ be compact, then
	$\Omega_\Sigma(\Xg)$ is closed. If $\fg$ is continuous and $\Sigma$ is uniformly ultimately bounded from~$\Xg$, then $\Omega_\Sigma(\Xg)$ is nonempty, compact, uniformly attractive for $\Sigma$ from $\Xg$, and it is the smallest (in the sense of inclusion) closed set with this latter property. If in addition $\Omega_\Sigma(\Xg)\subset \Xg$, then $\Omega_\Sigma(\Xg)$ is also forward invariant. 
\end{proposition}


\section{The Framework}\label{sec.Framework}
In this section we provide a number of new definitions on which the forthcoming analysis is based. In Section~\ref{sec.Properties}, we give a precise definition of \emph{steady-state property}. In Section~\ref{sec.Perturbations}, we give a unifying   definition of \emph{perturbation} of the plant dynamics and of its set of initial conditions. Finally, in Section~\ref{sec.Robustness}, we define the notion of \emph{robust regulator} relative to a given class of perturbations and to a given desired asymptotic property. Examples are given at the end of each section to help clarifying the proposed definitions.
\subsection{Steady-State Properties}\label{sec.Properties}
Consider  a system  of the form 
\begin{equation}\label{s.Sigma.ss_property}
\Sigma:\ \left\{ \begin{array}{lcl}
\dotxg &=& \fg(\xg)\\
\yg &=& \hg(\xg)
\end{array} \right.
\end{equation}
and a compact subset $\Xg\subset\R^\dxg$ of initial conditions. Under the assumptions of Proposition \ref{prop.Omega}, every complete trajectory of $\Sigma$ originating in $\Xg$ converges asymptotically, and uniformly,  to $\Omega_\Sigma(\Xg)$, which is compact and nonempty. 
In view of the uniform attractiveness  of $\Omega_\Sigma(\Xg)$ from $\Xg$, and in particular of the fact that it is the smallest closed set with such property, the trajectories of $\Sigma$ originating inside $\Omega_\Sigma(\Xg)$ have the usual interpretation as \emph{limiting trajectories} of the solutions of $\Sigma$ originating in $\Xg$. For this reason,  $\Omega_\Sigma(\Xg)$ is also referred to as the \emph{steady-state locus} of $(\Sigma,\Xg)$, and the elements of $\cS_\Sigma(\Omega_\Sigma(\Xg))$ as its \emph{steady-state trajectories}. 
This motivates the following definitions.
\begin{definition}[Steady-State Trajectories]
Given a pair $(\Sigma,\Xg)$, in which $\Sigma$ is a system of the form \eqref{s.Sigma.ss_property} and $\Xg\subset\R^\dxg$ a set, the elements of $\cS_\Sigma(\Omega_\Sigma(\Xg))$ are called the \emph{steady-state trajectories} of $(\Sigma,\Xg)$.
\end{definition}

For the sake of compactness, we denote by
\begin{equation*}
\SST_\Sigma(\Xi):=\cS_\Sigma(\Omega_\Sigma(\Xi))
\end{equation*}
the set of steady-state trajectories of $(\Sigma,\Xi)$.
\begin{definition}[Steady-State Property]\label{def.steady-state}
	A \emph{steady-state property} $\P$ on $(\Sigma,\Xg)$ is a statement $\P(\Sigma,\cO_\Sigma(\Xi))$ on the set of the steady-state trajectories of $(\Sigma,\Xg)$. In particular, we say that the steady-state trajectories of $(\Sigma,\Xi)$  \emph{enjoy} $\cP$ (or, simply, $(\Sigma,\Xi)$ enjoys $\cP$) if $\P(\Sigma,\cO_\Sigma(\Xi))$ holds true.
\end{definition}

The argument $\Sigma$ in  $\cP(\Sigma,\cO_\Sigma(\Xg))$ is introduced to allow  making statements about quantities that are ``system dependent''. For instance, all the statements about the output $\yg=\hg(\xg)$ make use of the system-dependent quantity~$\hg$. 
When $\Sigma$ and $\Xi$ are clear from the context, they are omitted, and we simply write $\Omega$, $\cS$, and $\cP(\cO)$ in place of $\Omega_{\Sigma}(\Xg)$, $\cS_{\Sigma}(\Xg)$ and $\cP(\Sigma,\cO_\Sigma(\Xg))$.
Some relevant steady-state properties are defined below.
\begin{example}[Set-Membership Property]\label{ex.P_A}
For a given set $A\subset\R^\dxg$,  the \emph{set-membership}  property is defined as $\P_A  := \propdef{\forall \xi\in\cO,\ \forall t\in \dom\xg,\ \xg(t)\in A}$.   
\end{example}
\begin{example}[Equilibrium Property]\label{ex.P_eq}
	The \emph{equilibrium} property is defined as $\P_{\rm eq}  := \propdef{\forall \xi\in\cO,\ \forall t,s\in \dom\xg,\ \xg(t)=\xg(s)}$. If $\P_{\rm eq}$ holds, then each steady-state trajectory is constant in time.
\end{example} 
\begin{example}[Regulation Properties]\label{ex.regulation_properties}
Let $\Sigma$ be obtained as the interconnection of a plant \eqref{s.wxy}, \eqref{d.e} with a controller \eqref{s.xc} (in this case $\xg=(w,\xp,x_c)$). 
Then, the \emph{asymptotic regulation} property is defined as  $\P_0 :=\propdef{\forall\xg\in\cO,\ \forall t\in\dom\xg,\  e(t):=h_e(w(t),\xp(t))=0}$.
Similarly, the approximate regulation property   \eqref{b.ls_e_vep} is represented by the steady-state property $\P_\vep  :=\propdef{\forall\xg\in\cO,\ \forall t\in\dom\xg,\ \big|h_e(w(t),\xp(t))\big| \le \vep}$.	
\end{example}
\begin{example}[Zero-Mean Property]
The following property, instead, characterizes a steady-state regulation error  $e$ with null DC component, without constraining its amplitude: 
\begin{equation*}
\cP_{\rm DC} :=\propdef{\forall\xg\in\cO,\ \int_{\dom \xg} h_e(w(t),\xp(t)) dt =0}.
\end{equation*}
\end{example}

Other more specific steady-state properties  will be introduced in the forthcoming sections. 

 \subsection{Perturbations}\label{sec.Perturbations}
For compactness, here and in the rest of the paper we let
\begin{subequations}\label{s.x}
	\begin{equation}
			\xw:=(w,\xp),
	\end{equation}
$\dxw=\dw+\dxp$, $\fw(\xw,u):=(s(w),\fp(w,\xp,u))$, $\hw(\xw):=\hp(w,\xp)$ and we rewrite \eqref{s.wxy} as
\begin{equation}
\Sigma_\xw :\ \left\{ \begin{array}{lcl}
\dotxw&=&\fw(\xw,u)\\
\yp &=& \hw(\xw).
\end{array}\right.
\end{equation} 
\end{subequations}
System \eqref{s.x} is referred to as the \emph{extended plant}. 
We further let $\Fw:=(\fw,\hw)$, and we suppose that it belongs to a given set  of functions $\cF$. For a given input $u$, the set of solutions of each  system  \eqref{s.x}   is then completely defined by the specification of a   point $(\Fw,\Xw)$ in the product space $\cF\x\cX$, being $X\in\cX$ a set of initial conditions for \eqref{s.x}, and  $\cX$  a subset of the power set of~$\R^\dxw$.   Since our aim is to study the asymptotic behavior of these solutions,   Proposition~\ref{prop.Omega} suggests to  restrict  the attention to the case in which $\cF$ is a set of continuous functions, and $\cX$ contains only compact subsets of $\R^\dxw$. Thus, from now on we assume $\cF\subset\cC^0$ and $\cX\subset\cK(\R^\dxw)$.

We equip $\cF$ with a topology $\tau_\cF$, and $\cX$ with a topology~$\tau_\cX$, in this way turning $(\cF,\tau_\cF)$ and $(\cX,\tau_\cX)$ into topological spaces. Then,  we  endow the product space $\cF\x\cX$ with the \emph{product topology}  $\tau:=\tau_\cF\x\tau_\cX$, which turns   $(\cF\x\cX,\tau)$ itself into a topological space. The topology $\tau$ on $\cF\x\cX$ constitutes the basic mathematical structure that allows us to formally define the concept of \emph{perturbation} in the space $\cF\x\cX$. The open sets of $\tau$, indeed, provide a formal characterization of the \emph{qualitative} meaning of ``vicinity'' between two different points of  $\cF\x\cX$, in the same way as norms usually do for vectors in $\R^n$. 
Different choices of $\tau_\cF$ and $\tau_\cX$ may be used to capture different ideas of ``variation'' of the function $\Fw$ and the initialization set $\Xw$ (see the examples below). 

With the topological space $(\cF\x\cX,\tau)$ defined, and with $(\Fw,\Xw)\in\cF\x\cX$, we say that a set $\MB{\cN}\subset\cF\x\cX$ is a \emph{neighborhood of $(\Fw,\Xw)$} if it contains an open set $U\in\tau$ containing $(\Fw,\Xw)$. We then call the elements  of $\MB{\cN}$ the $\MB{\cN}$-\emph{perturbations} (or simply \emph{perturbations}) of $(\Fw,\Xw)$.
Some examples of possible choices of $\tau$ are given below.
\begin{example}[Trivial and Discrete Topologies]\label{ex.topo1}
When $\tau_\cF$ and $\tau_\cX$ are the \emph{trivial topologies} on  $\cF$ and $\cX$ respectively, meaning that $\tau_\cF:=\{\emptyset,\,\cF\}$ and $\tau_\cX:=\{\emptyset,\,\cX\}$, then $\tau$ is the \emph{trivial topology} on $\cF\x\cX$, i.e. $\tau=\{\emptyset,\,\cF\x\cX\}$. The trivial topology is  the coarsest topology that can be defined on $\cF\x\cX$, and in fact it carries no information as each element of $\cF\x\cX$ has only one $\tau$-neighborhood given by $\cF\x\cX$ itself. On the contrary, if $\tau_\cF$ and $\tau_\cX$ are the \emph{discrete topologies} on $\cF$ and $\cX$, then $\tau$ is the discrete topology on $\cF\x\cX$ (i.e. $\tau$ coincides with the power set of $\cF\x\cX$) and it is the finest topology that can be defined on $\cF\x\cX$. In this topology, each singleton $\{(\Fw,\Xw)\}$ is a neighborhood of its element $(\Fw,\Xw)$.
\end{example}
\begin{example}[Hausdorff Topology]\label{ex.Hausdorff}
As $\cX$ contains only compact sets, a  natural choice for $\tau_\cX$ is the \emph{Hausdorff topology}, i.e. the topology induced by the \emph{Hausdorff distance} $d_{\rm H}$, which, for each two compact sets $X,\ Z\in\cX$, is defined as
\begin{equation*}
d_{\rm H}(X,Z) := \max\left\{ \sup_{x\in X}|x|_{Z},\ \sup_{z\in Z} |z|_X  \right\}.
\end{equation*}
In this topology, a set $Z$ is a perturbation of $X$ if there exists $\epsilon>0$ such that $X\subset Z+\epsilon\ball$ and $Z\subset X+\epsilon\ball$, with $\epsilon$ that \emph{quantifies} the entity of the perturbation.
\end{example}
\begin{example}[Projection Topologies]
In many cases, a different topology on $\cX$ may be preferred. For instance, suppose we aim to characterize variations of the sets of initial conditions involving only the $i$-th component $\xw_i$ of the initial state $\xw\in\Xw$. In this case  the Hausdorff topology is not suitable, as $d_{\rm H}$ weights uniformly variations in every directions. This can be rather  achieved  by letting $\tau_\cX$ be the topology generated\footnote{That is, the smallest topology including such sets.}  by the   sets
\begin{equation*}
\MB{\cN}(\Xw,\epsilon) := \big\{ Z\in\cX\st |z_i-\xw_i|<\epsilon,\,\forall \xw\in \Xw,\, \forall z\in Z \big\}
\end{equation*}
obtained by letting $\Xw$ and $\epsilon$ range in $\cX$ and $(0,\infty)$ respectively. The topology $\tau_\cX$ is not metrizable in this case, although the function $d_1(X,Z):=\sup_{(x,z)\in X\x Z} |\xw_i-z_i|$ is a semimetric on $\cX$. We also remark that $\tau_\cX$ coincides with the \emph{initial topology}\footnote{We recall that, given a family $\{f_\alpha\}_\alpha$ of functions $f_\alpha:A\to (B_\alpha,\tau_{B_\alpha})$, with $(B_\alpha,\tau_{B_\alpha})$  topological spaces, the initial topology of $f_\alpha$ on $A$ is  the coarsest topology for which each $f_\alpha$ is continuous.} of the projection map $\xw \mapsto \xw_i$.
\end{example}
\begin{example}[Weak $\cC^k$ Topology]\label{ex.weak_topology}
Let $\Pb\in\cK(\R^\dxw\x\R^\du)$ be an arbitrarily large compact set and, with $k\in\N$, let $\cF\subset\cC^k$ and define the semimetric $d_{(k,\Pb)}$ on $\cF$ as
\begin{equation*}
d_{(k,\Pb)}(\Fw,G) := \max_{i=0,\dots,k} \sup_{\pb\in\Pb}\left|\Fw^{(i)}(\pb)-G^{(i)}(\pb)\right|,
\end{equation*}		 
The topology $\tau_\cF$  induced by $d_{(k,\Pb)}$ is called the \emph{weak $\cC^k$ topology} \cite{Hirsch1994}. We observe that $d_{(k,\Pb)}$ is a metric on the space formed by the restrictions on $\Pb$ of the elements of $\cF$. The adjective ``weak'' refers to the fact that we restricted the ``$\sup$'' on the (pre-specified) compact set $\Pb$. By making $\Pb$ vary with~$\epsilon$, i.e. by considering the initial topology on $\cF$ induced by the family of functions $\{d_{(k,\Pb)}\}_{\Pb\in\cK(\R^\dxw\x\R^\du)}$, we obtain a much finer topology called the \emph{strong $\cC^k$} (or \emph{compact-open}) topology. Nevertheless, this strong version does not enjoy   some useful properties (e.g. it is not (semi-)metrizable) that its weaker version does, and hence in this paper we will only consider the weak version.
\end{example}
\begin{example}[Topology of Parameters Perturbations]\label{ex.perturbations.parametric}
	Let   $(P,\tau_P)$ be a topological space (called the \emph{parameter space}), and let $\cF$ be a set of $\cC^0$ functions indexed by $P$ (i.e. $\cF=\{\Fw_p\}_{p\in P}$ with $\Fw_p\in\cC^0$). 
Then $\cF$ models a family of  functions that is   \emph{parameterized} by the parameter $p\in P$. Typically $P$ is a subset of an Euclidean space endowed with the subset topology $\tau_P$.
By construction, there is a surjective map $\gamma:P\to\cF$ such that each $\Fw\in\cF$ is given by $\Fw=\gamma(p)$. We may assume that $\gamma$ is also injective (otherwise re-define $P$ by identifying the points yielding the same $\Fw$), so as $\gamma$ is invertible. We thus define the \emph{topology of parameter perturbations} $\tau_\cF$ on $\cF$ to be the initial topology of $\gamma\inv$, i.e. the  topology  generated by the sets $\gamma(U)$  for each open set $U$   of $\tau_P$. In the relevant case in which $P\subset\R^{n_p}$ for some $n_p\in\N$, and $\tau_P$ is induced by any norm on $\R^{n_p}$, then $\tau_\cF$ is   generated by the neighborhoods
\[
\MB{\cN}(\Fw,\epsilon) :=\left\{ G\in \cF\st  \left|p_\Fw -p_G \right|<\epsilon,\ p_G:=\gamma\inv(G) \right\} 
\]
for all $\Fw\in\cF$ and  $\epsilon>0$,
and in which $p_\Fw:=\gamma\inv(\Fw)$.
\end{example}
\begin{example}[Linear Perturbations]\label{ex.perturbations.linear}
\MB{Let $\cF$ be the set of \emph{linear maps}. Fix a basis for $\R^\dw$, $\R^\dxp$, $\R^\du$ and $\R^\dy$. Then, an invertible map $\gamma$ is defined that sends the \emph{matrix representation} $(M_s,M_{\fp},M_{\hp})\in \R^{\dw\x\dw}\x\R^{\dxp\x(\dw+\dxp+\du)}\x\R^{\dy\x(\dw+\dxp)}=:P$ of a given $\Fw=(s,\fp,\hp)\in\cF$ to $\Fw$ itself.} Then, this is a sub-case of   Example \ref{ex.perturbations.parametric}, and the topology of parameter perturbations induced on $\cF$ indeed coincides with the one induced  by any \emph{matrix norm} on the space $P$ of the matrix representations of $\cF$. Hence, the concept of variation captured by the topology of parameter perturbations coincides with the usual notion of parameter perturbations of linear systems, which is the one used in the context of structurally stable linear regulation~\cite{Francis1976,Davison1976,Byrnes1997}.  Moreover, in this case this topology also coincides with the  weak $\cC^0$ topology on $\cF$ with respect to any compact neighborhood $\Pb$ of the origin.
\end{example}

\subsection{Robustness}\label{sec.Robustness}
With $n_c\in\N$, $f_c\in\cC^0(\R^{\dc}\x\R^{\dy},\R^{\dc})$, $h_c\in\cC^0(\R^{\dc}\x\R^{\dy},\R^{\du})$ and $X_c\in\cK(\R^{\dc})$, consider a regulator of the form~\eqref{s.xc},  for convenience   rewritten hereafter
\begin{equation}\label{s.xc2}
	\Sigma_c:\ \left\{ \begin{array}{lclrl} 
	\dot x_c &=& f_c(x_c,y) &&x_c(0)\in  X_c \\
	u &=& h_c(x_c,y)
	\end{array} \right. 
\end{equation}
and consider the   interconnection between the extended  plant~\eqref{s.x} and  the  regulator   \eqref{s.xc2}, which reads as 
\begin{equation}\label{s.cl}
\Sigma_{\rm cl}(\Fw,\Xw):\ \left\{ \begin{array}{lclrl}
\dotxw &=& \fw(\xw,h_c(x_c,\hw(\xw))),&&x(0)\in \Xw\\
\dot x_c &=& f_c(x_c,\hw(\xw)),&&x_c(0)\in  X_c 
\end{array} \right.   
\end{equation}
in which we made explicit the dependence of $\Sigma_{\rm cl}$ from $F$ and~$X$.
The achievement by  regulator $\Sigma_c$  of a given steady-state property $\P$ for the closed-loop system $\Sigma_{\rm cl}(\Fw,\Xw)$  is the result of two subsequent goals: 
\begin{enumerate}
	\item\label{item.rob1} The existence of a non-empty steady state $\Omega_{\Sigma_{\rm cl}(\Fw,\Xw)}(X\x X_c)$ for  \eqref{s.cl} satisfying the properties of Proposition \ref{prop.Omega}.
	\item\label{item.rob2} The fulfillment of $\P$ by the closed-loop steady-state trajectories of $\Sigma_{\rm cl}(\Fw,\Xw)$.
\end{enumerate}
Clearly, Item \ref{item.rob2} is of interest only if  Item~\ref{item.rob1} is first ensured, as the achievement of property $\P$ is asked to the steady-state trajectories of $\Sigma_{\rm cl}(\Fw,\Xw)$, which exist only if Item~\ref{item.rob1} is first addressed. In turn, Item~\ref{item.rob1} can be seen as a \emph{stabilization requirement}, while Item~\ref{item.rob2} as a \emph{performance specification}.

Suppose that the regulator $\Sigma_c$ has been tuned under the assumption that the extended plant's data $(\Fw,\Xw)$   equal   a given \emph{nominal} value $(\nominal{\Fw},\nominal{\Xw})\in\cF\x\cX$. Then Items~\ref{item.rob1} and~\ref{item.rob2} above for the nominal case in which $(\Fw,\Xw)=(\nominal{\Fw},\nominal{\Xw})$ are formally  captured by the  following  definitions.
\begin{definition}[Nominal Stability]\label{def.S-nominal}
The regulator $\Sigma_c$  is said to be \emph{nominally stabilizing at $(\nominal{\Fw},\nominal{\Xw})\in\cF\x\cX$} if the system $\Sigma_{\rm cl}(\nominal{\Fw},\nominal{\Xw})$, given by \eqref{s.cl} for $(\Fw,\Xw)=(\nominal{\Fw},\nominal{\Xw})$,  is uniformly ultimately bounded from $\nominal{\Xw}\x X_c$.
\end{definition}
\begin{definition}[Nominal Steady-State Property]\label{def.P-nominal}
	The regulator $\Sigma_c$ is said to \emph{achieve the steady-state property $\P$ nominally  at $(\nominal{\Fw},\nominal{\Xw})\in\cF\x\cX$} (or to be \emph{$\P$-nominal}) if it is nominally stabilizing at $(\nominal{\Fw},\nominal{\Xw})$, 
and $(\Sigma_{\rm cl}(\nominal{\Fw},\nominal{\Xw}),\nominal{\Xw}\x X_c)$ enjoys $\P$  in the sense of
 Definition~\ref{def.steady-state}.
\end{definition}

As   the functions of \eqref{s.cl} are continuous,  then, by Proposition~\ref{prop.Omega}, a nominally stabilizing regulator guarantees that Item~\ref{item.rob1} above is fulfilled. Therefore, Definition \ref{def.P-nominal} is well posed.
Let now $\cF\x\cX$ be endowed with a topology $\tau$, as detailed in Section \ref{sec.Perturbations}, and suppose  that the same controller $\Sigma_c$, tuned on the nominal pair $(\nominal{\Fw},\nominal{\Xw})$, is applied to an extended plant \eqref{s.x} obtained by a possibly different pair $(\Fw,\Xw)$. Then, roughly speaking, the regulator $\Sigma_c$  will be called ``robust''  if the same nominal behavior expressed by Definitions \ref{def.S-nominal} and \ref{def.P-nominal} is maintained if the actual $(\Fw,\Xw)$ is  ``close-enough'' (relative to $\tau$) to $(\nominal{\Fw},\nominal{\Xw})$.
\begin{definition}[Robust Stability]\label{def.S-robust}
The regulator \eqref{s.xc2} is said to be \emph{robustly stabilizing at $(\nominal{\Fw},\nominal{\Xw})\in\cF\x\cX$ and with respect to $\tau$}   if there exists a $\tau$-neighborhood $\cN$ of $(\nominal{\Fw},\nominal{\Xw})$ such that, for each $(\Fw,\Xw)\in \cN$, the corresponding closed-loop system $\Sigma_{\rm cl}(\Fw,\Xw)$ given by \eqref{s.cl} is uniformly ultimately bounded from ${\Xw}\x X_c$.
\end{definition}

The $\tau$-neighborhood $\cN$ for which robust stability holds is called the \emph{robust stability neighborhood} of $(\nominal{\Fw},\nominal{\Xw})$. 
\begin{definition}[Robust Steady-State Property]\label{def.P-robust}
	The regulator~\eqref{s.xc2} is said to \emph{achieve the steady-state property~$\P$ robustly  at $(\nominal{\Fw},\nominal{\Xw})\in\cF\x\cX$ and with respect to $\tau$} (or to be \emph{$(\P,\tau)$-robust}) if it is robustly stabilizing at $(\nominal{\Fw},\nominal{\Xw})$ with respect to~$\tau$ and, by letting $\cN$ the  robust stability neighborhood of $(\nominal{\Fw},\nominal{\Xw})$, then the closed-loop system $(\Sigma_{\rm cl}(\Fw,\Xw),{\Xw}\x X_c)$ given by \eqref{s.cl} enjoys $\P$  in the sense of
	Definition \ref{def.steady-state}   for each $(\Fw,\Xw)\in \cN$.
	\end{definition}
\begin{remark}
One may  be concerned with the lack of \emph{uniformity} on the ultimate bound in Definition \ref{def.S-robust}, in the sense that, in this form, the definition admits the possibility of ``horizon-escaping'' phenomena of the limit sets as $(\Fw,\Xw)$ approaches the frontier of $\cN$. In fact, nothing in Definition \ref{def.S-robust} prevents the existence of a net $\{ (\Fw_\alpha,\Xw_\alpha) \}_\alpha$ with values in $\cN$ such that the corresponding net $\{ \Omega_{\Fw_\alpha}(\Xw_\alpha\x X_c) \}_\alpha$ leaves every compact subset of $\R^\dxw\x\R^\dc$. However, we underline that uniform boundedness  of the limit sets can be  enforced by opportunely modifying the property~$\P$. For instance, one may consider robustness relative to a  property of the kind
\begin{equation*}
\P'  := \P \wedge \propdef{  \cup_{\xi\in\SST}\range\xi \subset K  }
\end{equation*}
with $\wedge$ denoting the logical ``and'', and $K\subset\R^\dxw\x\R^\dc$ a given pre-specified bounded set, possibly dependent on $(\nominal{\Fw},\nominal{\Xw})$. 	 
\end{remark}

Three examples of robustness obtained for different steady-state properties are given below. The discussion of robustness with respect to output regulation properties instead is postponed to the forthcoming dedicated sections.
\begin{example}[Trivial and Universal Robustness]
	We consider two degenerate cases: 
(i) if $\tau$ is the discrete topology (Example \ref{ex.topo1})  then,  for any $\P$, every $\P$-nominal regulator is also $(\P,\tau)$-robust, in fact  $\{(\nominal{\Fw},\nominal{\Xw})\}$ is a $\tau$-neighborhood of $(\nominal{\Fw},\nominal{\Xw})$; (ii)   if $\tau$ is the trivial topology then, for every $\P$, no $\P$-nominal regulator is $(\P,\tau)$-robust, unless it is \emph{universal} (i.e. it achieves $\P$ for \emph{every} $(\Fw,\Xw)\in\cF\x\cX$). In fact, in this case, $\cF\x\cX$ is the only neighborhood of $(\nominal{\Fw},\nominal{\Xw})$.
\end{example}
\begin{example}[Asymptotic Stability Is Robust Invariance]
For a given $(\nominal{\Fw},\nominal{\Xw})\in\cF\x\cX$, suppose that the regulator $x_c$ is chosen so that $A:=\nominal{X}\x\R^\dc$ is invariant for $\Sigma_{\rm cl}(\nominal{F},\nominal{X})$ (in this case $\nominal{X}$ may represent, for instance, a desired equilibrium or operating set for $x$). Let $\P_A$ be the set-membership property of Example \ref{ex.P_A}. Then, by construction, the regulator $x_c$ is $\P_A$-nominal at $(\nominal{\Xw},\nominal{\Fw})$. Let $\tau=\tau_\cF\x\tau_\cX$, with $\tau_\cF$ the discrete topology and $\tau_\cX$ the Hausdorff topology. Then, a regulator is $(\P_A,\tau)$-robust at $(\nominal{\Xw},\nominal{\Fw})$ if and only if $A$ is   \emph{(locally) asymptotically stable} for $\Sigma_{\rm cl}(\nominal{F},\R^\dxw)$. We remark that this claim is only valid because $\tau_\cF$ is the discrete topology  and, thus robustness may refer only to perturbations of $\nominal{\Xw}$. Sufficiency is obvious, while necessity follows by observing that, if  $x_c$ is $(\P_A,\tau)$-robust, then  exists $\mu>0$ such that $A$ is uniformly attractive from $(\nominal{\Xw}+\mu\ball)\x\R^\dc$. Thus, in view of \cite[Proposition 7.5]{Goebel2012book}, since $A\subset (\nominal{\Xw}+\mu\ball)\x\R^\dc$ is invariant, it is also asymptotically stable. 
\end{example}
\begin{example}[Total Stability]\label{ex.total_stab} Let $\cF=\cC^1$ and $\cX=\cK(\R^{\dx})$. Suppose that $(\nominal\Fw,\nominal\Xw)$ is such that \eqref{s.cl} has an equilibrium which is locally exponentially stable with a domain of attraction $\cD$ including $\nominal\Xw\x X_c$. Let $\tau_\cF$ be the weak $\cC^1$ topology (Example~\ref{ex.weak_topology}) on a compact set $\Pb:=\Xb\x\Ub$ with $\cD\subset \Xb$, and let $\tau_\cX$ be the Hausdorff topology on $\cX$. Then for small enough perturbations $(\Fw,\Xw)$ of $(\nominal\Fw,\nominal\Xw)$, the system $\Sigma_{\rm cl}(\Fw,\Xw)$ still has an equilibrium, possibly different from the nominal one, which is still locally exponentially stable  \cite[Lemma~5]{Astolfi2017}. This is known as \emph{total stability} and, in the language of this paper, it is equivalent to $(\P_{\rm eq},\tau_{\cF}\x\tau_\cX)$-robustness where $\P_{\rm eq}$ is the equilibrium property of Example~\ref{ex.P_eq}.
\end{example}


%
%

\section{Regulators with Linear Internal Models}\label{sec.Linear}
In this section we focus on a particular class of output-feedback regulators \eqref{s.xc2} embedding a \emph{linear internal model~(LIM)}. In addition to the   linear setting of~\cite{Francis1975,Francis1976,Davison1976}, linear internal models have been extensively used also in the nonlinear literature~(see, e.g.~\cite{Isidori1990,Huang1994,Byrnes1997,Serrani2001,Byrnes2003,Astolfi2015,Astolfi2017}). 

From now on we consider an extended plant  \eqref{s.x} with a regulation error \eqref{d.e}, and we    assume the following.
\begin{assumption}\label{ass.e_y}
The regulation error $e$ is directly measured from $y$, i.e. $y=(e,\yaux)$ where $\yaux\in\R^\dyaux$, $\dyaux:=\dy-\de$, is an ``auxiliary output''.
\end{assumption}

Assumption \ref{ass.e_y} is not in principle necessary, although the slightly weaker concept of \emph{readability}\footnote{The regulation error $e$ is said to be readable from the output $y$, if there exists a matrix $Q\in\R^{\de\x\dy}$ such that $e=Qy$ \cite{Francis1975}.} of $e$ from $y$ is indeed necessary to achieve robust asymptotic regulation  for linear systems \cite[Proposition 2]{Francis1975}, \cite[Condition~5 in Lemma~1]{Davison1976}.

Furthermore, in this section we focus on ``smooth variations'' of the  function $F=(f,h)$ of the extended plant \eqref{s.x}. More precisely, we suppose that $\cF\subset\cC^1$ and, with $\Xb\subset\R^\dx$ and $\Ub\subset\R^\du$  arbitrarily large compact neighborhoods,  we let $\Pb:=\Xb\x\Ub$  and we endow $\cF$ with the subset topology $\tau_\cF$ induced by the weak $\cC^1$  topology $\tau_{\cC^1}$ defined in Example \ref{ex.weak_topology}. Regarding the initialization set, instead, we let $\cX=\cK(\R^\dx)$.

Regulators embedding a linear internal model  are systems of the form \eqref{s.xc2} in  which, possibly after a change of coordinates, the state $x_c$   is partitioned as
\begin{subequations}\label{s.xc_linear_im}
	\begin{equation}
		 x_c = \big( \cim,\, \cst \big)
	\end{equation}
with $\cim\in\R^\dcim$ and $\cst\in\R^\dcst$, being $\dcim$, $\dcst\in\N$ such that $\dcim+\dcst=\dc$, and with the maps $f_c$ and $h_c$ making~\eqref{s.xc2}   read as follows
\begin{equation}
\Sigma_c^{\rm LIM}:\left\{\begin{array}{lcl}
\dotcim &=& \Acim \cim + \Bcim e\\
\dotcst &=& \fcst(\cst,\cim,y)\\
u &=& h_c(\cst,\cim,y) 
\end{array}\right. \quad x_c(0)\in X_c,
\end{equation}
\end{subequations} 
in which $\Acim$ and $\Bcim$ are linear maps such that the subsystem~$\cim$ is controllable from $e$, and $\Acim$ contains the same modes that we would like to reject from $e$ at the steady state (a precise characterization is given in the forthcoming sections).
In the output regulation literature, the structure \eqref{s.xc_linear_im} is said to be of the \emph{post-processing} type \cite{Astolfi2013,Bin2019}. The subsystem $\cim$ is the internal model unit, and it is responsible of generating the right feedforward control action   ideally keeping $e(t)=0$ at the steady state. The subsystem $\cst$, instead, typically has the role of stabilizing the whole closed-loop system. 

In the remainder of the section we consider different settings in which \eqref{s.xc_linear_im} is used, and we characterize   some of its robustness properties relative to asymptotic and approximate regulation properties.

\subsection{Robustness of the Linear Regulator}\label{sec.lin.lin}
As a first case, we consider the canonical linear setting of~\cite{Francis1975,Francis1976,Davison1976}, in which the nominal extended plant's function~$\nominal{F}$ is linear, the nominal exosystem's   function $\nominal{s}$ is marginally stable, known and  unperturbed, and the canonical Linear Regulator is used \cite{Davison1976}.
In particular, with $\cF_L\subset\cC^1$   the set of     linear functions, we suppose that $\nominal{\Fw}\in\cF_L$, and
we 
 let
\begin{equation}\label{lim.d.cF_linear_sknown}
\cF:= \big\{ F\in\cF_L\st s=s^\circ\big\}.
\end{equation} 
Namely the only variations $F$ of $\nominal{F}$ that we consider are those for which $F$ remains linear and $\nominal{s}$ is untouched.
The topology $\tau_\cF$ that we define on $\cF$  is the subset  $\cC^1$ topology  and it   coincides with the usual topology of  parametric perturbations of the matrix representation of $F$ (see Example~\ref{ex.perturbations.linear}). 

The control goal that we consider in this section is \emph{global} robust asymptotic regulation. Namely, we aim at finding a regulator ensuring   $\lim_{t\to\infty}e(t)=0$ from every initial condition even when applied to an extended plant whose function $F$ slightly differs from $\nominal{F}$ (although remaining in~$\cF$). In the language of this paper, we thus seek a {$\cP_0$-robust} regulator, where
\begin{equation}\label{d.P0}
	\cP_0:=\propdef{\forall\xi\in\cO,\ \forall t\in\dom\xi,\ h_e(w(t),\xp(t))=0}.
\end{equation} 
As the aimed result is global in the initial conditions, we let $\nominal{X}\in\cX$ be arbitrary, and we let $\tau_\cX$ be the trivial topology  $\tau_\cX:=\{\emptyset,\cX\}$. In fact, this implies that every robust regulator necessarily achieves the objective globally in the initial conditions, since the only $\tau_\cX$-neighborhood of $\nominal X$ is the whole $\cX$.

It is a classical result in linear control theory,  that this control objective can be guaranteed by means of a regulator of the form \eqref{s.xc_linear_im} constructed as follows \cite{Davison1976}:
\begin{enumerate}[label=\alph*)]
	\item $X_c$ is arbitrary (e.g. $(\cim(0),\cst(0))=0$).
	\item $\Acim$ and $\Bcim$ are chosen so that $(\Acim,\Bcim)$ is controllable, $\Phi$ is marginally stable, and the characteristic polynomial of $\Phi$ coincides with the minimal polynomial of  $s^\circ$.
   \item $\fcst$ and $h_c$ are linear functions chosen to stabilize the overall closed-loop system. This choice is always possible under standard detectability and stabilizability conditions on the plant, and provided that the following \emph{non-resonance condition} holds  \cite[Lemma~14]{Davison1976}
   \begin{equation*}
   \rank \begin{pmatrix}
   \partial \nominal{f_p}(0)/\partial \xp  -\lambda I & \partial \nominal{f_p}(0)/\partial u\\
   \partial \nominal{h_e}(0)/\partial \xp & 0
   \end{pmatrix} = \dxp+\de, 
   \end{equation*} 
   for all $\lambda\in\sigma(\nominal{s})$, in which we recall $f_p$ and $h_p$ are the functions defining, respectively, the plant's dynamics~\eqref{s.wxy} and the regulation error~\eqref{d.e}.
\end{enumerate}
The regulator constructed in this way is known as the \emph{Linear Regulator}, and it enjoys the following robustness property.
 \begin{theorem}\label{thm:linear_robustness}
 	The Linear Regulator  is $(\cP_0,\tau_\cF\x\tau_\cX)$-robust at $(\nominal F,\nominal X)$, with $\cP_0$ defined   in \eqref{d.P0}.
 \end{theorem}
 \begin{proof}
 	The fact that the Linear Regulator is robustly stabilizing the closed-loop system \eqref{s.x}, \eqref{s.xc_linear_im} at $(\nominal F, \nominal X)$ with respect to $\tau_\cF\x\tau_\cX$  is a direct consequence of the definition of $\tau_\cF$ and the continuity of the spectrum of the (linear) closed-loop  function. The fact that the steady-state trajectories of the perturbed closed-loop system enjoy $\cP_0$ follows from  Theorem~\ref{thm.P_nu_weak} and Remark~\ref{rmk.ap} (see below, Section~\ref{sec.lin.qT}),   once noted that such trajectories are necessarily almost periodic (see, e.g., \cite[Theorem 4.2]{CorduneanuAlmostP})
 \end{proof}

An intuitive argument for the formidable robustness result enunciated by Theorem \ref{thm:linear_robustness}   is the following:  no matter how large are the variations of $F$ with respect to $\nominal F$, if $F$ remains linear, then  the closed-loop system still consists of a stable linear system driven by the same  exosystem. Therefore, the closed-loop   steady-state trajectories keep oscillating at the same frequencies, and the linear internal model of \eqref{s.xc_linear_im} is still able to generate the  error-zeroing control action needed in the perturbed case. This ``immersion'' property is at the basis of most of the existing robustness results for nonlinear systems~\cite{Huang1994,Huang1995,Byrnes1997,Isidori2012,Astolfi2017,Bin2019}. In particular, it is at the basis of the well-known \emph{integral action} \cite{Astolfi2017}: if $\Phi=0$ in~\eqref{s.xc_linear_im},  and if the rest of regulator ensures local exponential stability of the controlled plant when $w=0$, then for small constant $w(t)$ and small $\cC^0$ perturbations of the plant's dynamics, the closed-loop system still has a stable equilibrium (see Example~\ref{ex.total_stab}), and asymptotic regulation is achieved.    

Nevertheless, this  immersion  property is hardly satisfies in a general nonlinear setting under unstructured  perturbations (see Theorem~\ref{thm.controesempio} in Section~\ref{sec.nonlinear}).

\subsection{Nonlinear Perturbations and (Weak) Periodic Robustness}\label{sec.lin.T}

In this section we allow general $\cC^1$ perturbations of $F$, thus including the case in which the perturbed function $F$ may be \emph{nonlinear} and the perturbation affects also the exosystem map~$s$. In particular, we let $\cX=\cK(\R^\dx)$ and, as  in  \cite{Astolfi2015,Astolfi2017}, we let $\cF$ be an arbitrary subset of $\cC^1$. We endow $\cX$ with an arbitrary topology $\tau_\cX$, and we let $\tau_\cF$ be the subset topology induced by the weak $\cC^1$ topology (Example~\ref{ex.weak_topology}) on an arbitrary compact neighborhood $\Pb\subset\R^\dx\x\R^\du$ of the origin.

In this setting, a robustness result of the kind given in Theorem~\ref{thm:linear_robustness} is no more possible in general. Nevertheless, along the lines of~\cite{Astolfi2015}, we can show that a regulator of the kind~\eqref{s.xc_linear_im}, embedding a suitably designed linear internal model, can still guarantee robustness of a weaker, approximate regulation objective  consisting in the rejection from $e$ of the harmonics included in the internal model dynamics. We treat here the periodic case in which the frequencies included in the internal models are multiple of a fundamental one. We postpone the more general case of arbitrary frequencies to the next section. 

As our aim is to highlight the role of the linear internal model in a nonlinear setting, we do not restrict the possible exosystem dynamics and we do not fix the other parts  of the regulator \eqref{s.xc_linear_im} (i.e. the   maps $\fcst$ and $h_c$), which  unlike the previous case can be nonlinear. We rather give a general robustness result (Theorem~\ref{thm.P_T_weak}) which is independent from their specific choice, provided that some basic robust stability properties hold. This permits to   separate the contributions of the exosystem and of the subsystems $\cim$ and $\cst$ of~\eqref{s.xc_linear_im} in terms of their effect on the steady-state trajectories of the closed-loop system. Later in the section, we support Theorem~\ref{thm.P_T_weak} with two other results giving conditions on the exosystem and on $(\fcst,h_c)$ ensuring that its assumptions  are fulfilled.

 In the remainder of the section, we consider a regulator of the form \eqref{s.xc_linear_im}, in which  $\Phi$ and $G$ are chosen as follows: we first fix an arbitrary period $T>0$, and an arbitrary number $d\in\N$ of harmonics to reject. Then, we choose $\dcim$, $\Phi$ and $G$ in such a way that
\begin{enumerate}[label=IM-\Alph*),ref=IM-\Alph*,leftmargin=1.3cm]
	\item\label{item.IM.dcim} $\dcim:=(2d+1)\de$.
	\item\label{item.IM.Phi} The spectrum of $\Phi$ is
	\begin{equation*}
	\sigma(\Phi) = \{0\}\cup \left( \bigcup_{k=1}^d \left\{i\dfrac{2\pi k}{T},\,-i\dfrac{2\pi k}{T}\right\} \right) 
	\end{equation*} 
	in which $i$ denotes the imaginary unit, and in which each eigenvalue has algebraic and geometric multiplicity  $\de$.
	\item\label{item.IM.G} $(\Phi,G)$ is controllable.
\end{enumerate}
Item \ref{item.IM.Phi}, in particular, implies that the unforced internal model subsystem $\cim$ can generate all the $T$-periodic signals having a non-zero bias and the first $d$ harmonics starting from the fundamental frequency $1/T$. 

We now characterize the robustness properties of any   regulator  \eqref{s.xc_linear_im} embedding the internal model defined by Items \ref{item.IM.dcim}, \ref{item.IM.Phi} and \ref{item.IM.G} above, in terms of asymptotic rejection from the steady-state regulation error $e$ of a bias and the harmonics at frequencies $k/T$, $k=1,\dots,d$. For $(F,X)\in\cF\x\cX$, we denote by $\Sigma_{\rm cl}(F,X)$ the closed-loop system  composed by the extended plant \eqref{s.x}, with initial conditions in $X$, and the regulator \eqref{s.xc_linear_im} with the  internal model unit constructed above. We denote by $\xi:=(w,\xp,\cim,\cst)$ the overall state. Then, for a given continuous function $\alpha:\R\to\R^m$, $m\in\N$, we define the Fourier coefficients 
\begin{equation*}
\fourier_k(\alpha) := \int_0^T \alpha(t)e^{-i 2\pi k t/T} dt 
\end{equation*}
 and we let
\begin{equation*}
\cQ_d^m:=\Big\{ \alpha:\R\to\R^m\st c_k(\alpha)=0,\, k=0,\dots,d\Big\}
\end{equation*}
be the subspace of the functions $\R\to\R^m$ that have null Fourier coefficient $\fourier_k(\alpha)$  for all $k=0,\dots, d$. 
Then, we define the following steady-state property
\begin{equation}\label{d.PTweak}
	\begin{aligned}
 {\cP_{T,weak}}  := \propdef{\forall & \xi\in\cO_{\Sigma_{\rm cl}(F,X)}(X\x X_c),\\&  \cim \text{ is not }T\text{-periodic or }  e \in  \cQ_d^{\de}}.
\end{aligned}
\end{equation}
If $\cim$ is \emph{not} $T$-periodic, then Property $\cP_{T,weak}$ is always satisfied. When, however, $\cim$ is $T$-periodic, $\cP_{T,weak}$ asks that the steady-state regulation error $e$ has zero mean value and zero amplitude at every frequency $k/T$,  $k=1,\dots, d$.  
Then, with $(\nominal F,\nominal X)\in\cF\x\cX$   the nominal value of the extended plant's data, the following result holds.
\begin{theorem}\label{thm.P_T_weak}
	Consider a regulator $\Sigma_c^{\rm LIM}$ of the form \eqref{s.xc_linear_im}, with $\dcim$, $\Phi$ and $G$ chosen according to Items \ref{item.IM.dcim}, \ref{item.IM.Phi} and \ref{item.IM.G}. Suppose that $\Sigma_c^{\rm LIM}$ is robustly stabilizing at $(\nominal F,\nominal X)$ with respect to $\tau_\cX\x\tau_{\cF}$. Then the regulator  is $(\cP_{T,weak},\tau_\cX\x\tau_{\cF})$-robust at $(\nominal F,\nominal X)$, with $\cP_{T,weak}$ defined as in \eqref{d.PTweak}.
\end{theorem}

Theorem \ref{thm.P_T_weak}, whose proof is reported in Appendix \ref{apd.proof.Thm_PT_weak}, characterizes the effect of the internal model defined by Items~\ref{item.IM.dcim}, \ref{item.IM.Phi} and \ref{item.IM.G}  independently on the rest of the regulator. In particular, if the remaining degrees of freedom $(\dcst,\fcst,h_c,X_c)$ of \eqref{s.xc_linear_im} can be chosen to ensure   robust stabilization of the closed-loop system, and that the steady-state trajectories of $\cim$ are $T$-periodic, then rejection from $e$ of the harmonics embedded in the internal model holds robustly. 
	
We also remark that $T$-periodicity of the steady-state trajectories of $\cim$ is a condition which is also exosystem-dependent, and the ability to design $(\dcst,\fcst,h_c,X_c)$ to guarantee   robust stabilization highly depends on the particular   extended plant considered.
For instance, if the extended plant is linear, then $(\fcst,h_c)$ can be chosen as a simple linear stabilizer, thus reducing to a particular case of the Linear Regulator of Section~\ref{sec.lin.lin}. If the extended plant is nonlinear, instead, only few cases are covered in the literature. For instance, in~\cite{Byrnes2004,Marconi2007,Astolfi2013,WanIsiLiuSu,Bin2019}, semiglobal  solutions based on ``high-gain arguments'' are proposed for classes of   minimum-phase normal forms, and in \cite{Astolfi2015,Astolfi2017}  \emph{forwarding} techniques are used for  a class of non-necessarily minimum-phase  systems in general form and for ``small $w$''. Under suitable conditions, and if the solutions of the exosystem are $T$-periodic, some of these design solutions yield regulators of the form \eqref{s.xc_linear_im} which are robustly stabilizing at $(\nominal F,\nominal X)$ and that also ensure that~$\cim$ is $T$-periodic at the steady state. Therefore, they strengthen the result of Theorem \ref{thm.P_T_weak} to  robustness with respect to the following steady-state property
\begin{equation}\label{d.PT}
\cP_T := \propdef{\forall\xi\in\cO_{\Sigma_{\rm cl}(F,X)}(X\x X_c),\  e \in \cQ_d^{n_e} }
\end{equation}
which represents the ``strong'' version of $\cP_{T,weak}$.

\subsection{Achieving Strong Periodic Robustness}\label{sec.lin.sT}

In this section, we further investigate the problem of individuating sufficient conditions under which the claim of Theorem~\ref{thm.P_T_weak} may be strengthen to $\cP_{T}$-robustness, where $\P_T$ is defined in \eqref{d.PT}. 

First, we consider a \emph{local} (in the initial conditions) result. In this case, we let $\cF=\cF_w\x \cF_{p}$, where $\cF_w$ is a set of~$\cC^1$ functions $\R^\dw\to\R^\dw$ and $\cF_{p}$ is a set of $\cC^2$ functions of the form $(\fp,\hp)$, with $\fp:\R^\dw\x\R^\dxp\x\R^\du\to\R^\dxp$ and $\hp:\R^\dw\x\R^\dxp\to\R^\dy$. We endow $\cF_w$ with an arbitrary topology~$\tau_{\cF_w}$ and, with $\Xb\subset\R^\dw\x\R^\dxp$ and $\Ub\subset\R^\du$ arbitrarily large compact neighborhoods of the respective origins, we  endow $\cF_{p}$ with   the subset topology~$\tau_{\cF_p}$ induced by the weak~$\cC^2$ topology on $\Pb:=\Xb\x\Ub$ (see Example \ref{ex.weak_topology}). 
We then let $\tau_\cF:=\tau_{\cF_w}\x\tau_{\cF_p}$. This allows us to consider ``independent variations'' of $s$ and $(\fp,\hp)$, and thus to better distinguish the assumptions on the exosystem from that on the controlled plant. Finally, we let $\cX =\cK(\R^\dx)$, and we endow it with the Hausdorff topology~$\tau_{\cX}$ (see Example \ref{ex.Hausdorff}). 

We look at the closed-loop system $\Sigma_{\rm cl}(F,X)$, given by interconnection between the extended plant \eqref{s.x} and the regulator~\eqref{s.xc_linear_im}, as the cascade of the exosystem
 	\begin{equation}\label{s.w_lin}
	\dot w=s(w),
\end{equation}
into the controlled plant 
\begin{equation}\label{s.old_z} 
\Sigma_{\rm pc} :\,\left\{ \begin{array}{lcl}
	\dotxp &=& \fp(w,\xp,h_c(\cst,\cim,\hp(w,\xp)))\\
	\dotcim &=& \Phi\cim+Gh_e(w,\xp)\\
	\dotcst &=& \fcst(\cst,\cim,\hp(w,\xp)).
\end{array}\right. 
\end{equation}
 Let $\nominal F=(\nominal s,\nominal\fp,\nominal \hp)\in\cF$  denote the nominal extended plant function.   Then, we make the following assumption.
 \begin{assumption}\label{ass.lin_LES}
 The following hold:
 \begin{enumerate}
 	\item\label{item.AssLES.w} 	There exists a $\tau_{\cF_w}$-neighborhood $\cN_w$ of $\nominal s$  and, for every $\epsilon>0$, a $\delta(\epsilon)>0$, such that every solution to \eqref{s.w_lin} with  $s\in\cN_w$ and
 	satisfying $|w(0)|\le \delta(\epsilon)$ also satisfies $|w(t)|\le \epsilon$ for all $t\ge 0$. Moreover,   the solutions to \eqref{s.w_lin} with $s\in\cN_w$  satisfy
 	\begin{equation*}
 		\lim_{t\to\infty}|w(t+T)-w(t)|=0
 	\end{equation*}
 	uniformly\footnote{
 		That is, if for every $\epsilon>0$ and every $W\in\cK(\R^\dw)$, there exists $r>0$ such that every solution $w$ to~\eqref{s.w_lin}   originating in $W$ satisfies $|w(t+T)-w(t)|\le \epsilon$ for all $t\ge r$.}
 	
 	\item\label{item.AssLES.c} The triple $(\dcim,\Phi,G)$ is chosen according to Items \ref{item.IM.dcim}, \ref{item.IM.Phi} and \ref{item.IM.G}, $X_c$ is compact, the functions $\fcst$ and $h_c$ are $\cC^2$, and the system  $\Sigma_{\rm pc}$ with $w=0$ and with $(\fp,\hp)=(\nominal\fp,\nominal\hp)$ nominal
 	is locally exponentially stable with a domain of attraction including $\nominal{X_p}\x X_c$.
 \end{enumerate}  
 \end{assumption}

Item \ref{item.AssLES.w} of Assumption~\ref{ass.lin_LES} consists of two parts. The first is a   marginal stability requirement on the origin of the exosystem state-space which is \emph{uniform} in the perturbations of the function $s$ inside $\cN_w$, in the sense that the scalar $\delta(\epsilon)$ is, for fixed $\epsilon$, the same for all $s\in\cN_w$. The second requires that the solutions of the exosystem with $s\in\cN_{w}$ are asymptotically $T$-periodic, uniformly over compact subsets of   initial conditions.
Item~\ref{item.AssLES.c}, instead, requires   the \emph{nominal} controlled plant \eqref{s.old_z} to be  locally exponentially stable when $w=0$.
 This, in turn, can be seen as a design requirement for the remaining parts $(\dcst,\fcst,h_c,X_c)$ of the regulator, which have to be designed to locally stabilize the plant when $w=0$. 
 Under this assumptions, the following result holds.
\begin{proposition}\label{prop.linear_LES}
Suppose that Assumption \ref{ass.lin_LES} holds. 
Then the regulator \eqref{s.xc_linear_im} is $(\cP_T,\tau_\cF\x\tau_\cX)$-robust at $(\nominal F,\{0\})$, with $\cP_{T}$ defined in \eqref{d.PT}.
\end{proposition}

The claim of Proposition~\ref{prop.linear_LES}, whose proof is reported in Appendix \ref{apd.proof.prop.LES}, is \emph{local} in the initial condition, in that the nominal initialization set for the extended plant is $\nominal X=\{0\}$. In fact, being $\tau_\cX$ the Hausdorff topology, any element of any $\tau_\cX$-neighborhood of $\nominal X$  contains an open ball around the origin (see Example \ref{ex.Hausdorff}). In this sense, the result of Proposition~\ref{prop.linear_LES} generalizes the design philosophy of \cite{Astolfi2015}, in which local asymptotic stability of the controlled plant is achieved by forwarding techniques.

We  consider now a non-local extension of the result of  Proposition~\ref{prop.linear_LES}.
We let $\cF$   be an arbitrary subset of $\cC^1$, which we endow with the subset topology $\tau_\cF$ induced by the weak $\cC^1$ topology on the compact set $\Pb:=\Xb\x\Ub$, with $\Xb\subset\R^\dw\x\R^\dxp$ and $\Ub\subset\R^\du$ arbitrarily large compact neighborhoods of the origin. Then,    we let $\cX=\cK(\R^\dx)$, and we endow it with an arbitrary topology $\tau_\cX$. We let $\tau:=\tau_\cF\x\tau_\cX$. 
With  $\nominal F\in\cF$   the nominal extended plant's  function and $\nominal X\in\cX$   the nominal set of initial conditions, we assume the following.
\begin{assumption}\label{ass.lin2}
There exists $\tau$-neighborhood $\cN$ of~$(\nominal F,\nominal X)$ such that, for each $(F,X)\in\cN$, the following hold:
\begin{enumerate}
	\item\label{ass.lin2.item1} The closed-loop system \eqref{s.w_lin}-\eqref{s.old_z} is uniformly ultimately bounded from $X\x X_c$.
	
	\item\label{ass.lin2.item2} The solutions $w$ to \eqref{s.w_lin} originating in $W:=\{w\in\R^\dw\st (w,\xp)\in X\}$ satisfy
	\begin{equation*}
		\lim_{t\to\infty}|w(t+T)-w(t)|=0
	\end{equation*}
	uniformly.
	
	\item\label{ass.lin2.item3} The controlled plant \eqref{s.old_z} is \emph{incrementally input-to-state stable} in the sense of \cite[Definition 4.1]{Angeli2002} on $X\x X_c$, and  with respect to the input $w$. Namely, there exist $\beta\in\KL$, and $\rho\in\K$ such that, for every two solutions $(w,\xp,x_c)$ and $(w',\xp',x_c')$ of the closed-loop system~\eqref{s.w_lin}-\eqref{s.old_z} originating in $X\x X_c$ the following holds
	\begin{equation*}
		\begin{aligned}
			|(\xp(t),&x_c(t)) - (\xp'(t),x_c'(t))|\\&\le \beta\big(|(\xp(0),x_c(0)) - (\xp'(0),x_c'(0))|,t\big)\\
			&\qquad + \rho\Big( \textstyle\sup_{s\in[0,t)}|w(s)-w'(s)|\Big)
		\end{aligned}
	\end{equation*}
for all $t\ge 0$.
\end{enumerate}  
\end{assumption}

Assumption~\ref{ass.lin2} extends to a non-local setting the basic properties implied locally by the previous Assumption~\ref{ass.lin_LES}. In addition to uniform ultimate boundedness of the closed-loop system, Assumption~\ref{ass.lin2} requires the perturbed controlled plant  to satisfy a non-local \emph{incremental input-to-state stability}  property with respect to the exogenous signal $w$~\cite{Angeli2002}. We remark that the same property can be also characterized in terms of \emph{convergent systems} \cite{ruffer_convergent_2013}, and we refer the reader to~\cite{Angeli2002,Pavlov2006,ruffer_convergent_2013,giaccagli_incremental_2020} and the references therein for control design techniques ensuring that such property holds.

 Under Assumption~\ref{ass.lin2}, the following result holds.
\begin{proposition}\label{prop.T_robustness_nonlocal}
	Suppose that Assumption~\ref{ass.lin2} holds.  Then the regulator~\eqref{s.xc_linear_im} is $(\cP_{T},\tau)$-robust at $(\nominal F,\nominal X)$, with $\cP_{T}$ defined as in \eqref{d.PT}.
\end{proposition}
\begin{skproof} 
The fact that the regulator is robustly stabilizing at $(\nominal F,\nominal X)$ follows by Item~\ref{ass.lin2.item1} of Assumption~\ref{ass.lin2}. 
Therefore, in view of  Theorem~\ref{thm.P_T_weak}, it suffices to show that the steady-state trajectories of the closed-loop system are $T$-periodic. This, in turn, follows by means of the same arguments used in the proof of Proposition~\ref{prop.linear_LES} (see Appendix~\ref{apd.proof.prop.LES}), due to uniformity of the limit $\lim_{t\to\infty}|w(t+T)-w(t)|=0$ asked by Item~\ref{ass.lin2.item2} of Assumption~\ref{ass.lin2}, and since Item~\ref{ass.lin2.item3} implies that for every two solutions $(w,\xp,x_c)$ and $(w',\xp',x_c')$ of the closed-loop system~\eqref{s.w_lin}-\eqref{s.old_z} originating in $X\x X_c$, $\limsup_{t\to\infty}|(\xp(t),x_c(t)) - (\xp'(t),x_c'(t))|\le \rho( \limsup_{t\to\infty}|w(s)-w'(s)|)$ holds uniformly \cite[Section~IV]{Angeli2002}.
\end{skproof}

\subsection{Robustness of Arbitrary Harmonics Rejection}\label{sec.lin.qT}
The results of Sections~\ref{sec.lin.T} and~\ref{sec.lin.sT} apply to the case in which the frequencies included in the internal model unit of \eqref{s.xc_linear_im} are multiple of a fundamental frequency $1/T$. These results may be extended to the case in which the frequencies to reject are arbitrary, in the context of generalized Fourier coefficients. In particular, in this section we provide an extension (Theorem~\ref{thm.P_nu_weak} below) of   Theorem~\ref{thm.P_T_weak}.

With $d\in\N$, let $\{\nu_k\}_{k=1}^d$ be the set of frequencies to be rejected from the steady-state regulation error $e$, and   consider a regulator of the form~\eqref{s.xc_linear_im}  in which the internal model triple $(\dcim,\Phi,G)$ is chosen as follows:
\begin{enumerate}[label=IM-\Alph*'),ref=IM-\Alph*',leftmargin=1.3cm]
	\item\label{item.IM2.dcim} $\dcim:=(2d+1)\de$.
	\item\label{item.IM2.Phi} The spectrum of $\Phi$ is
	\begin{equation*}
	\sigma(\Phi) = \{0\}\cup \left( \bigcup_{k=1}^d \left\{i2\pi\nu_k,\,-i2\pi\nu_k\right\} \right) 
	\end{equation*} 
	in which $i$ denotes the imaginary unit, and in which each eigenvalue has algebraic and geometric multiplicity  $\de$.
	\item\label{item.IM2.G} $(\Phi,G)$ is controllable.
\end{enumerate} 

Let, for convenience, $\nu_0:=0$  and, for every $m\in\N$ and every   function $\alpha:\R\to\R^m$, define the (generalized) Fourier coefficients as
\begin{equation*}
c_k'(\alpha) := \lim_{T\to\infty} \dfrac{1}{T}\int_0^T \alpha(t) e^{-i2\pi\nu_k t}dt ,
\end{equation*}
whenever   they exist. 
\begin{remark}\label{rmk.ap}
	A relevant case in which $c_k'(\alpha)$ exists for all $k$ and $\nu_k$ is when $\alpha$ is \emph{almost periodic} (see, e.g., \cite[Section~I.3]{CorduneanuAlmostP}).
\end{remark}

Define the sets
\begin{align*}
\cL_d^m{}'  &:= \Big\{ \alpha:\R\to\R^m\st c_k'(\alpha) \text{ exists } \forall k=0,\dots,d\Big\}\\
\cQ_d^m{}'&:=\Big\{  \alpha \in\cL_d^m{}' \st c_k'(\alpha)=0\  \forall k=0,\dots,d \Big\}.
\end{align*}
Proceeding as before, we let
\begin{equation}\label{d.Pnu_weak}
\P_{\nu,weak} : = \propdef{\eta\notin\cL_d^{\dcim}{}' \text{ or } e\in\cQ_d^\de{}'}.
\end{equation}

As  in Sections \ref{sec.lin.T} and \ref{sec.lin.sT}, we let $\cF\subset\cC^1$, $\cX=\cK(\R^\dx)$,  $\tau_\cF$ be the subset weak $\cC^1$ topology defined in Example~\ref{ex.weak_topology} on an arbitrary compact subset $\Pb$ of $\R^{\dx}\x\R^{\du}$, and $\tau_\cX$ be an arbitrary topology on $\cX$. Then, with   $(\nominal F,\nominal X)\in\cF\x\cX$   the nominal value of the extended plant's data and  the following result holds.
\begin{theorem} \label{thm.P_nu_weak}
	Consider a regulator $\Sigma_c^{\rm LIM}$ of the form \eqref{s.xc_linear_im}, with $(\dcim,\Phi,G)$ chosen according to Items \ref{item.IM2.dcim}, \ref{item.IM2.Phi} and \ref{item.IM2.G}. Suppose that $\Sigma_c^{\rm LIM}$ is robustly stabilizing at $(\nominal F,\nominal X)$ with respect to $\tau:=\tau_\cF\x\tau_\cX$. Then the regulator  is $(\cP_{\nu,weak},\tau)$-robust at $(\nominal F,\nominal X)$.
\end{theorem}
\begin{skproof} 
	The proof follows from the same argument of the proof of Theorem~\ref{thm.P_T_weak} (see Appendix~\ref{apd.proof.Thm_PT_weak}) once noted that, for each $n=1,\dots,2d+1$ and each $k=0,\dots,d$, we have
	\begin{align*}
	c_k'&\left(\eta_{1}^{(n)}\right) = \lim_{T\to\infty}\dfrac{1}{T}  \int_0^T \eta_{1}^{(n)}(t)e^{-i2\pi\nu_k t}dt \\&= \lim_{T\to\infty} \dfrac{1}{T}\left[ \eta_{1}^{(n-1)}(t) e^{-i2\pi \nu_k t}\right]_0^T + i2\pi\nu_k  c_k'\left(\eta_{1}^{(n-1)}\right)\\
	&= i2\pi\nu_k  c_k'\left(\eta_{1}^{(n-1)}\right)
	\end{align*}
	in which the integrals exist as long as $\eta\in\cL_d^\dcim{}'$,   the   term  $\lim_{T\to\infty}  (1/T) [ \eta_{1}^{(n-1)}(t) e^{-i2\pi \nu_k t} ]_0^T$  vanishes since $\eta_{1}^{(n-1)}(t)e^{-i2\pi\nu_k t}$ is bounded, and   $i2\pi\nu_k\in\sigma(\Phi)$.
\end{skproof}

\section{Regulators with Nonlinear Internal Models}\label{sec.nonlinear}

In the previous sections we have shown that, for a class of nonlinear problems, a regulator embedding a linear internal model is able to guarantee  robust harmonic rejection from the steady-state regulation error. Nevertheless, results concerning asymptotic regulation have been only given      in the linear case when $s$ is not uncertain. In this section,  we consider  regulators of a general form \eqref{s.xc2}, in which also the   internal model is allowed to be nonlinear. The only constraints we consider are the finite-dimensionality of the state space and smoothness of the vector fields $f_c$ and $h_c$. This, indeed, guarantees the existence of the limit set as detailed in Proposition~\ref{prop.Omega}. 

In the context of minimum-phase single-input-single-output normal forms with unitary relative degree\footnote{We consider this setting because it represents the case in which the most general results exist.}, we  show that, while it is true that a nonlinear regulator always exists ensuring nominal asymptotic regulation, robust asymptotic regulation is instead \emph{impossible} in general, at least in the relevant case of~$\cC^0$ perturbations. In particular, we show that there exist (very simple) systems for which there can not exist a smooth finite-dimensional regulator   ensuring robust asymptotic regulation. For simplicity, and since  it represents the most general available existence result, we shall restrict the discussion to the design of \cite{Marconi2007} (recalled hereafter)  in the context of   single-input-single-output minimum-phase normal forms. We remark, however,  that the same conclusions apply to the regulator in~\cite{Byrnes2004} and \cite{Chen2005}, as well as to all their numerous extensions.

We consider systems of the form \eqref{s.wxy} with   $\dy=\de=1$, $y=e$, with the plant's state $\xp$ which admits the decomposition $\xp:=(\zeta,e)$, $\zeta\in\R^{\dxp-1}$,   with $(\fp,\hp)$ such that, in certain coordinates, the plant's  equations read  as follows
\begin{equation}\label{1n:e:r_eq:nl_ze}
\begin{array}{lcl}
\dot \zeta&=& \vhi(w,\zeta,e)\\
\dot e&=& q(w,\zeta,e) + b(w,\zeta,e)u
\end{array}
\end{equation}
and, finally, with $s$, $\vhi$, $q$ and $b$ that satisfy the following assumption.
\begin{assumption}\label{ass.sicon}
The functions $(s,\vhi,q,b)$ are\footnote{  This assumption can be weakened (see \cite{Marconi2007}), but it is assumed for simplicity.} $\cC^\infty$  and the following properties hold:
\begin{enumerate}
	\item There exists $\und b>0$ such that $b(w,\zeta,e)\ge \und b$ for all $(w,\zeta,e)\in\R^\dw\x\R^{\dxp}$. 
	\item There exists a compact set $\cA\subset\R^\dw\x\R^{\dxp-1}$ which is locally asymptotically stable for the system
	\begin{align*}
	\dot w &= s(w), & \dot\zeta &=\vhi(w,\zeta,0)
	\end{align*}
	with an open domain of attraction $\cD\supset \cA$.
\end{enumerate} 
\end{assumption}  

We let $\cF$ be the set of smooth functions $F=(s,\fp,\hp)$ with the above properties and satisfying Assumption~\ref{ass.sicon}. Moreover, we let $\nominal F\in\cF$, and $\nominal X$ be any compact set such that, for some arbitrary bounded set $E\subset\R$, $\nominal X\subset\cD\x E$. Then, with $\cP_0$ the asymptotic regulation property defined in Example~\ref{ex.regulation_properties},  the result of \cite{Marconi2007} can be stated, in our setting, as follows.
\begin{theorem}\label{thm.sicon}
There always exists a finite-dimensional regulator of the form \eqref{s.xc2}  with $(f_c,h_c)\in\cC^0$ which achieves the steady-state property  $\cP_0$ nominally at $(\nominal F,\nominal X)$.
\end{theorem}

Theorem \ref{thm.sicon} gives an affirmative answer to the question whether or not we can always find a regulator embedding a nonlinear internal model that can effectively ensure asymptotic regulation in the nominal case. In general, possibly under  additional assumptions, the design of \cite{Marconi2007}, as well as those of \cite{Byrnes2004}, \cite{Chen2005} and the related extensions, can also guarantee   semi-global  or global   approximate regulation property robustly (i.e. the steady-state property $\cP_{\vep}$ defined in Example~\ref{ex.regulation_properties}), when the functions perturbations are meant in the $\cC^0$ topology. 

The question of whether or not robust asymptotic regulation might be ensured by a smooth finite-dimensional regulation has, however, a negative answer in general when perturbations are meant in the $\cC^0$ topology. In particular, we let $\cR$ be the class of problems obtained with:
\begin{enumerate}
	\item An extended plant of the form~\eqref{s.wxy},~\eqref{d.e}, for some   $\dw,\dxp,\dy,\du,\de\in\N$, satisfying    Assumption~\ref{ass.e_y}.
	\item  $F:=(s,\fp,\hp)\in\cF:=\cC^0$, where $\cF$ is given the weak~$\cC^0$ topology on    the compact set $\P:=\Xb\x\Ub$, where  $\Xb\subset\R^\dx$ ($\dx:=\dw+\dxp$) and $\Ub\subset\R^\du$ are arbitrary compact neighborhoods of the respective origin (see Example~\ref{ex.weak_topology}).
	\item $X\in\cX:=\cK(\R^\dx)$, where $\cX$ is given an arbitrary topology $\tau_\cX$.
\end{enumerate}
Then, the following result holds.
\begin{theorem}\label{thm.controesempio}
	There exist problems in $\cR$ and nominal data  $(\nominal F,\nominal X)\in\cF\x\cX$ for which   no   finite-dimensional regulator of the form \eqref{s.xc2}  with $(f_c,h_c)\in\cC^1$ exists that  achieves the steady-state property $\cP_0$ robustly at $(\nominal F,\nominal X)$  with respect to $\tau_\cF\x\tau_\cX$.
\end{theorem}

The proof of Theorem~\ref{thm.controesempio} (presented at the end of the section) is by counterexample. In particular, we consider a problem in~$\cR$ obtained with $\dw=2$, $\dxp=1$, $\dy=\de=\du=1$, $\nominal X=\nominal W\x \nominal X_p$, with $\nominal W:=\{(0,1)\}$ and $\nominal X_p\in\cK(\R)$ any, and with the extended plant satisfying the following equations
\begin{equation}\label{s.ce.wx}
\begin{aligned} 
\dot w &=  \begin{pmatrix}
0 & 1\\
-1 & 0
\end{pmatrix}w& w(0)&=(0,1)\\
\dotxp &= w_1+u, & \xp(0)&\in\nominal X_p\\
e &= \xp
\end{aligned}
\end{equation}
 which is a special case of \eqref{s.ex.wxy}.
  With $N\in\N$, let    
	\begin{align*}
	\cT_N \! = \!\Big\{ \signal:\R\to\R\!\st \!&\signal(t)= \alpha  + \sum_{n=1}^N \big( \beta^n\sin(nt) \!+\!\gamma^n\cos(nt) \big),\\& \alpha ,\beta^n,\gamma^n\in\R   \Big\}.
\end{align*}
be the vector space of all the time signals obtained as the linear composition of $N$ harmonics and a bias.  Then, the   mapping $\phi_N:\cT_N\to\R^{2N+1}$ defined as
\begin{equation*}
\signal\mapsto  \phi_N(\signal) :=(\alpha,\beta^1,\gamma^1,\dots,\beta^N,\gamma^N).
\end{equation*}
 is an isomorphism.  With $\epsilon\ball(\R^{2N+1})$   the ball of radius $\epsilon>0$ in $\R^{2N+1}$, we define the \emph{ball of radius $\epsilon$ in $\cT_N$} as
\begin{equation*}
\epsilon\ball(\cT_N) :=  \phi_N\inv\left(\epsilon\ball(\R^{2N+1})\right)=\Big\{ \signal\in\cT_N\st |\phi_N(\signal)|<\epsilon \Big\} .
\end{equation*}
Then, the following holds.
\begin{proposition}\label{prop.xi}
	Suppose that, for some $N\in\N$ and $\epsilon>0$, there exist   $m\in\N$, $g\in\cC^1(\R^m,\R^m)$, $\theta:\cC^1(\R^m,\R)$, and a system of the form
	\begin{align}\label{s:xi} 
	\dot \chi &=  g(\chi) &
	v &= \theta(\chi) , & \chi \in\R^m
	\end{align}
	such that  for every $\signal\in\epsilon\ball(\cT_N)$ there exists   $\chi_0\in\R^m$ such that the (unique) solution $\chi(t)$ of \eqref{s:xi} originating at $\chi(0)=\chi_0$ satisfies
	$v(t)=\theta(\chi(t))=\signal(t)$ for all $t\in\Rplus$. Then $m\ge 2N+1$.
\end{proposition}

The proof of Proposition \ref{prop.xi} is in Appendix~\ref{apd.proof.prop.xi}. Next, we have the following.
\begin{lemma}\label{lem.sin}
	Let $K\subset\R^2$ be a compact set including $0$. 
	For each $\epsilon>0$ and each $N\in\N$ there exists $\delta>0$ such that, for every $\signal\in\delta\ball(\cT_N)$ there exists $c_\signal\in\cC^0(\R^2,\R)$ satisfying $\sup_{k\in K}|c_\signal(k)|<\epsilon$ such that $c_\signal(\sin(t),\cos(t)) = \signal(t)$.
\end{lemma}

The proof of Lemma \ref{lem.sin} is presented in Appendix~\ref{apd.proof.lem.sin}. We are now ready to prove Theorem~\ref{thm.controesempio}.

\begin{proof}[Proof of Theorem~\ref{thm.controesempio}]
Suppose that a regulator of the form \eqref{s.xc2}, with $(f_c,h_c)\in\cC^1$ exists that achieves  the steady-state property $\cP_0$ robustly at $(\nominal F,\nominal X)$ and with respect to $\tau_\cF\x\tau_\cX$ when applied to the extended plant \eqref{s.ce.wx}. Then, the regulator is robustly stabilizing at $(\nominal F,\nominal X)$   with respect to $\tau_\cF\x\tau_\cX$, and there exists a $\tau_\cF$ neighborhood $\cN_\cF$ and a $\tau_\cX$ neighborhood of $\nominal X$ such that, for every $(F,X)\in\cN_\cF\x\cN_\cX$ the steady-state trajectories $(w,\xp,x_c)\in\cO_{\Sigma_{\rm cl}(F,X)}(X\x X_c)$ of the closed-loop system $\Sigma_{\rm cl}(F,X)$ given by the interconnection of the regulator \eqref{s.xc2} with \eqref{s.ce.wx} satisfy $e(t)=\xp(t)=0$ for all $t\ge 0$. This, implies $\dot e(t)   = 0$ for almost all $t\ge 0$. 

Pick a set $\tilde F_\epsilon\subset\cF$ such that, for all $F\in\tilde F_\epsilon$,  $s=\nominal s$, $\hp=\nominal \hp$ and $\fp$ is such that  
\begin{equation*}
\dotxp = q(w) + w_1+u,
\end{equation*}
with $q\in\cC^0$ such that $\sup_{(w,\xp)\in\Xb}|q(w)|<\epsilon$. Then we can find $\epsilon\sr>0$ such that $\tilde F_{\epsilon\sr} \subset \cN_\cF$. 

Pick $X=\nominal X$, which lies in $\cN_\cX$ by definition. Then, for all $F\in\tilde F_\epsilon$, every solution in $\cO_{\Sigma_{\rm cl}(F,X)}(X\x X_c)$ satisfies $w(t) = (\sin(t),\cos(t))$ for all $t\ge 0$. Moreover, since the map $(w,\xp,x_c)\mapsto q(w)+w_1+h_c(x_c,\xp)$ is continuous, then $\dot e(t)   = 0$ for almost all $t\ge 0$ implies $q(w(t)) =-  w_1(t)-h_c(x_c(t),\xp(t))$ for all $t\ge 0$. Therefore, with $m=\dw+n_c$, $\chi:=(w,x_c)$, $g(\chi):=(\nominal s(w),f_c(x_c,0))$ and $\theta(\chi):=-w_1-h_c(x_c,0)$, the system
	\begin{align}\label{s:xi2} 
\dot \chi &=  g(\chi) &
v &= \theta(\chi)  
\end{align}
necessarily has the property that   there exists $\Xi_{\epsilon\sr}\subset \nominal W\x X_c$ such that, for each $F\in\tilde F_{\epsilon\sr}$ there exists $\chi_0\in\Xi_{\epsilon\sr}$ such that, the unique solution to \eqref{s:xi2} originating at $\chi_0$ satisfies  $\theta(\chi(t))=q(\sin(t),\cos(t))$ for all $t\ge 0$.  

Now, use Lemma \ref{lem.sin} to find, for $K=\{w\in\R^2\st (w,\xp)\in \Xb, \xp\in\R\}$,  $\epsilon=\epsilon\sr$ and $N>(m-1)/2$, a $\delta$ such that for every $\signal\in\delta\ball (\cT_N)$ there exists $c_\signal\in \cC^0$ satisfying $\sup_{w\in K} |c_\signal(w)|\le \epsilon\sr$ and $c_\signal(\sin(t),\cos(t))=\signal(t)$. 
Then, for all such $\signal\in\delta\ball (\cT_N)$, taking $q=c_\signal$ produces an  $F \in\tilde{F}_{\epsilon\sr}$ for which $q(\sin(t),\cos(t))=\signal(t)$.  Therefore, the system~\eqref{s:xi2} has the property that, for every $\signal\in\delta\ball(\cT_N)$, there exists $\chi_0\in\Xi_{\epsilon\sr}$, such that the unique solution to \eqref{s:xi2} originating at $\chi_0$ satisfies  $\signal(t)  =\theta(\chi(t))$ for all $t\ge 0$. By Proposition~\ref{prop.xi}, however, this implies $m:=\de+n_c> 2N+1 > m$ which is a contradiction.
\end{proof}
\begin{remark}
The proof of Theorem~\ref{thm.controesempio} relies on the construction of an ad hoc set of perturbations of $F$ which makes impossible for the regulator to generate all the possible corresponding steady state control actions. We remark the in all the elements of this perturbation set,  the  function of the exosystem $s$  is kept equal to the nominal value $\nominal s$.  Thus, we are conceptually in the same setting of Theorem~\ref{thm:linear_robustness}, in which the extended plant is linear and only the plant is perturbed, with the only difference  that in this setting arbitrarily small nonlinear perturbations are allowed.
\end{remark}

	\begin{table*}
	\begin{center}
		\def\arraystretch{1.5}
		\begin{tabular}{ | m{2.8cm}  | m{1.4cm} | m{7.6cm} |  m{4.35cm} |}  
			\hline
			\textbf{Internal Model} & \textbf{Property} & $\boldsymbol{(\cF,\tau_\cF)}$ \textbf{(sub-classes of)} &  \textbf{References} \\
			\hline 
			\multirow[t]{4}{=}{Linear}    & 
			$\cP_{T}$, $\cP_\nu$  &  \begin{cellitemize}
				\item[$\cF$:] $\cC^1$ functions
				\item[$\tau_\cF$:] weak $\cC^1$ topology
			\end{cellitemize}        &  Sec. \ref{sec.lin.T}, \ref{sec.lin.sT}, \ref{sec.lin.qT}\newline   Ref.  \cite{Astolfi2015} \\
			\cline{2-4} &  \multirow[t]{3}{=}{ $\cP_0$ } & \begin{cellitemize}
				\item[$\cF$:] $s$ linear and nominal, $(f,h)$ polynomial of fixed order
				\item[$\tau_\cF$:] parameter perturbations
			\end{cellitemize}    
			&  Sec. \ref{sec.lin.T}, \ref{sec.lin.sT}, \ref{sec.lin.qT}\newline Ref. \cite{Huang1994,Byrnes1997,jie_huang_remarks_2001}\\
			\cline{3-4} & & \begin{cellitemize}
				\item[$\cF$:] $s$ nominal, $(s,f,h)$ linear
				\item[$\tau_\cF$:] parameter perturbations
			\end{cellitemize} & 
			Sec. \ref{sec.lin.T}, \ref{sec.lin.sT}, \ref{sec.lin.qT}\newline
			Ref. \cite{Francis1976,Davison1976}\\
			\cline{3-4} & & \begin{cellitemize}
				\item[$\cF$:] $s=0$, $(f,h)$ continuous, $f$ 0-LES
				\item[$\tau_\cF$:] weak $\cC^1$ topology
			\end{cellitemize} & Sec. \ref{sec.lin.T}, \ref{sec.lin.sT}, \ref{sec.lin.qT}\newline Ref.  \cite{Astolfi2017}\\
			\hline
			Nonlinear   & $\cP_{\varepsilon}$\newline{\footnotesize ($\varepsilon$ arbitrary)} & \begin{cellitemize}
				\item[$\cF$:] continuous functions
				\item[$\tau_\cF$:] weak $\cC^1$ topology
			\end{cellitemize} &   Sec. \ref{sec.nonlinear}\newline Ref. \cite{Byrnes2004,Marconi2007,Marconi2008,Isidori2012,Bin2019}\\
			\hline
			Linear Adaptive   & $\cP_{0}$ & \begin{cellitemize}
				\item[$\cF$:] linear functions
				\item[$\tau_\cF$:] parameter perturbations
			\end{cellitemize} &  Ref. \cite{Serrani2001,marino_output_2003,marino_output_2007,Bin2019adaptlinear}   \\\hline
			Nonlinear Adaptive${}^{\color{red} \star}$  & $\cP_{\varepsilon}$ & \begin{cellitemize}
				\item[$\cF$:] continuous functions
				\item[$\tau_\cF$:] weak $\cC^1$ topology
			\end{cellitemize} &   Ref. \cite{forte_robust_2017,Bin_classtype_2019,bin_approximate_2020,bernard_adaptive_2020}\\
			\hline
		\end{tabular}
	\end{center} 
	\caption{\normalfont{Overview of the robustness properties achievable by internal model based regulators. Specifications on $(\cX,\tau_\cX)$ are omitted because  context-dependent. Properties $\cP_0$ and $\cP_{\varepsilon}$ are defined in Example~\ref{ex.regulation_properties}, Property $\cP_{T}$ in \eqref{d.PT}, and Property $\cP_\nu$ is the strong version of \eqref{d.Pnu_weak}. The weak $\cC^1$ topology is defined in Example~\ref{ex.weak_topology}, while the topology of parameters perturbations in Example~\ref{ex.perturbations.parametric}. ${}^{\color{red} \star}$formal proof not yet given.}}
	\label{table}
\end{table*}

 \section{Conclusions}

In this paper, we have proposed a unifying framework where robustness of arbitrary steady-state properties with respect to arbitrary perturbations of the dynamics and initialization set	can be analyzed. In this setting, we have reinterpreted existing results and provided new results showing that the robustness property of the Linear Regulator extends to relevant classes of nonlinear problems if a milder harmonic rejection property -- rather than asymptotic regulation -- is pursued.  The given results are prescriptive, in the sense that they do not propose a specific design for the regulator, but rather they give general sufficient conditions serving as design guidelines. \MB{Table~\ref{table} summarizes the main achievable robustness properties  for different classes of systems and regulators.}

For what concerns asymptotic regulation, we have shown by counterexample that it cannot be achieved robustly with a smooth finite-dimensional  regulator, at least in the relevant case of arbitrary $\cC^0$ perturbations. \MB{Hence,   the properties of the linear regulator do not extend to nonlinear system in general. This result  gives an important answer to a long-standing question, and  points   us towards three main directions:}
\begin{enumerate}
	\item \MB{We left out discontinuous and hybrid controllers, e.g. sliding mode regulators \cite{shtessel_sliding_2014}. These have been already proved in several other domains to reach results otherwise unachievable with smooth controllers. The question whether we can recover robustness by means of hybrid control is therefore still open.}
	
	\item  Our results motivate  and theoretically support  the recent emergence of approximate regulation designs, and call for approaches using  run-time adaptation to automatically improve the steady-state regulation bound, in this way obtaining both robustness and performance.
	
	\item Our analysis relies on the finite dimensionality of the regulator's state space. Hence, robustness of asymptotic regulation is not ruled out if an \emph{infinite-dimensional} regulator is used, as in nonlinear \emph{repetitive control} schemes (see \cite{ghosh2000nonlinear,califano_stability_2018,astolfi_francis-wonham_2019,astolfi2022repetitive} and the references therein).
\end{enumerate}

\MB{As anticipated in the introduction, a proper treatise of discontinuous and hybrid plants or controllers  requires us to leave the context of differential equations. A possible setting is the framework of \cite{Goebel2012book}, where similar properties for limit sets as those used here can be recovered. Then, the notions introduced in this paper can be directly applied.}
\MB{Handling infinite-dimensional systems is instead more critical. If the state-space is reflexive (e.g., it is a Hilbert space), then we can rely on a weaker notion of limit sets, maintaining the definition given in Section~\ref{sec.Prelimiraies}, but with closure that has to be understood in the weak topology. For non-reflexive spaces, the extension is instead more critical.}
 Moreover, it is also worth mentioning that most of what is said in the paper can be extended to the case in which $w$ belongs to a class, say $\cW$, of functions not necessarily coincident with the set of solutions to a differential equation, provided that $\cW$ has some sequential closeness properties necessary to define the limit sets. 

Finally, we also underline that the implementation of infinite-dimensional \MB{or discontinuous} asymptotic regulators may anyway require a some approximation ruining, in practice,  asymptotic regulation. Therefore,  our results ultimately suggest that, in a nonlinear uncertain setting, approximate -- rather than asymptotic -- regulation is the correct way to approach the problem if robustness is sought.

\appendix
\subsection{Proof of Theorem \ref{thm.P_T_weak}} \label{apd.proof.Thm_PT_weak}
	If $\Sigma_c^{\rm LIM}$ is robustly stabilizing at $(\nominal F,\nominal X)$ with respect to $\tau$ in the sense of Definition \ref{def.S-robust}, there exists a $\tau$-neighborhood $\cN$ of $(\nominal F,\nominal X)$ such that, for each $(F,X)\in\cN$, $\cO_{\Sigma_{\rm cl}(F,X)}(X\x X_c)$ is well-defined and nonempty.
	
	Since $(\Phi,G)$ is controllable, we can pick a basis of the state space for which $\Phi$ and $G$ have the following form 
	\begin{align}\label{d.lin.Phi_G}
	\Phi &=\begin{pmatrix}
	0 & I & 0 &   \cdots & 0\\
	0 & 0 & I &  \cdots & 0\\
	\vdots & &&\ddots &  \vdots\\
	0 & &&  &  I\\
	-a_1 I & -a_2 I & \dots && -a_{2d+1}I
	\end{pmatrix}& G&=\begin{pmatrix}
	0\\0\\\vdots\\ 0\\I
	\end{pmatrix}
	\end{align} 
	in which $a_1,\dots, a_{2d+1}$ are the coefficients of the characteristic polynomial of $\Phi$, i.e. satisfying
	\begin{equation}\label{pf.e.char_p_Phi}
	\lambda^{2d+1} + a_{2d+1}\lambda^{2d}+\dots+ a_2\lambda + a_1 =0,\ \ \forall\lambda\in\sigma(\Phi).
	\end{equation}
	
	Pick $(F,X)\in \cN$ and let $(w,x,\cim,\cst)$ be any solution in $\cO_{\Sigma_{\rm cl}(F,X)}(X\x X_c)$. If $\cim$ is not $T$-periodic then $\cP_{T,weak}$ trivially holds, so we assume $\cim$ periodic. 
	
	Partition $\cim$ as $\cim=\col(\cim^1,\dots,\cim^{2d+1})$ with $\cim^\ell\in\R^{\de}$ for $\ell=1,\dots,2d+1$. 
	In view of \eqref{d.lin.Phi_G} this yields 
	\begin{equation*}
	\dotcim^{\ell}=\cim^{\ell+1} \qquad\forall \ell=1,\dots,2d.
	\end{equation*}
	Hence,  from the equation of $\dotcim^{2d+1}$, we obtain 
	\begin{equation}\label{pf.e.regr_eta_dot}
	(\cim^{1})^{(2d+1)} + a_{2d+1}(\cim^{1})^{(2d)}+\cdots+ a_2 \dotcim^{1} + a_1\cim^1 = e.
	\end{equation}
	Integrating by parts, for $k\in\N$ and $n=1,\dots,2d+1$ we obtain
	\begin{equation*}
	\begin{aligned}
	c_k&\left((\cim^{1})^{(n)}\right)  = \int_0^T (\cim^{1})^{(n)}(t) e^{-i2\pi k t/T} dt  
	\\&=\left[ (\cim^{1})^{(n-1)}(t)e^{-i2\pi k t/ T}\right]_0^T + \dfrac{i 2\pi k}{T}c_k\left((\cim^{1})^{(n-1)}\right)\\
	&=\dfrac{i 2\pi k}{T}c_k\left((\cim^{1})^{(n-1)}\right),
	\end{aligned}
	\end{equation*}
	in which we have used the fact that $(\cim^{1})^{(n-1)}(t)e^{-i2\pi k t/ T}$ is $T$-periodic.
	By induction on $n$ we thus obtain
	\begin{equation*}
	c_k\left((\cim^{1})^{(n)}\right)= (\lambda_k)^n c_k\left(\cim^{1}\right),
	\end{equation*}
	with $\lambda_k := i2\pi k/T$. From \eqref{pf.e.regr_eta_dot} we thus obtain
	\begin{equation*}
	c_k(e) = \Big( (\lambda_k)^{2d+1}+ a_{2d+1}(\lambda_k)^{2d}+\cdots+ a_{2}\lambda_k + a_1 \Big)c_k\left(\cim^{1}\right).
	\end{equation*}
	Since for $k\le d$, $\lambda_k\in\sigma(\Phi)$, then \eqref{pf.e.char_p_Phi} implies $c_k(e) = 0$, and the claim follows.  
	\ignorespaces\hfill$\QED$

\subsection{Proof of Proposition \ref{prop.linear_LES}} \label{apd.proof.prop.LES}
The following lemma can be proved by means of the same arguments of \cite[Theorem 10.3]{khalil_nonlinear_2002} (see also \cite{Angeli2002}).
\begin{lemma}\label{lem.tech_lemma}
Suppose that Assumption~\ref{ass.lin_LES} holds. Then there exists a $\tau_{\cF}\x\tau_\cX$-neighborhood $\cN$ of $(\nominal{F},\{0\})$ such that the system~\eqref{s.w_lin}-\eqref{s.old_z} obtained with $(F,X)\in\cN$ is uniformly ultimately bounded from $X\x X_c$, and all its solutions $\xb:=(w,\xp,x_c)$ originating in $X\x X_c$ satisfy $\lim_{t\to\infty}|\xb(t)-\xb(t+T)|=0$ uniformly. 
Namely, for each $(F,X)\in\cN$ and   $\vep>0$, there exists $r>0$ such that and every solution  $\xb$ to \eqref{s.w_lin}-\eqref{s.old_z} originating in $X\x X_c$ satisfies  $|\xb(t)-\xb(t-T)|\le\vep$ for all $t\ge r$.
\end{lemma}

Lemma~\ref{lem.tech_lemma} implies that all the steady-state solutions of $\Sigma_{\rm cl}(F,X)$ -- where  $\Sigma_{\rm cl}(F,X)$ denotes the closed-loop system~\eqref{s.w_lin}-\eqref{s.old_z} obtained for $(F,X)\in\cN$ and with $\cN$ given by Lemma~\ref{lem.tech_lemma} -- are $T$-periodic. In fact, since $(F,X)\in\cN$ implies $F\in\cC^1$, then, by Assumption~\ref{ass.lin_LES}, for every $(F,X)\in\cN$ the solutions of the closed-loop system $\Sigma_{\rm cl}(F,X)$ are unique, and  there exists a continuous function $\vhi:\R\x\R^{\dx+\dxc}\to \R^{\dx+\dxc}$ such that every solution $\xb:=(w,\xp,x_c)$ of $\Sigma_{\rm cl}(F,X)$   satisfies $\xb(t)=\vhi(s,\xb(t-s))$ for all $t\ge 0$ and $s\in[0,t]$. 

Pick arbitrarily $(F,X)\in\cN$, Then $\Omega_{\Sigma_{\rm cl}(F,X)}$ is non-empty in view of Lemma~\ref{lem.tech_lemma}. Pick arbitrarily $\xb\in\cO_{\Sigma_{\rm cl}(F,X)}(X\x X_c)$ and $t\ge 0$. Then,  there exist a sequence $\{\xb_n\}_n$ of solutions in $\cS_{\Sigma_{\rm cl}(F,X)}(X\x X_c)$ and a sequence $\{t_n\}_n$ of times $t_n$ satisfying $t_n\to\infty$ such that $\xb_n(t_n)\to \xb(t)$. Moreover, by continuity of $\vhi$, we have
\begin{equation*}
	\xb_n(t_n+T) = \vhi(T,\xb_n(t_n))\to \vhi(T,\xb(t)) = \xb(t+T).
\end{equation*}

 Then, in view of the uniformity of the convergence $\lim_{t\to\infty}|\xb(t)-\xb(t+T)|=0$ over compact subset established by  Lemma~\ref{lem.tech_lemma},  we conclude that   for every $\ell>0$  there exists $n\sr(\ell)>0$  such that, for all $n\ge n\sr(\ell)$ the following holds
\begin{align*}
	|&\xb(t+T)-\xb(t)| \\ &\le |\xb(t)-\xb_n(t_n)|+|\xb(t+T)-\xb_n(t_n+T)| \\& \qquad +|\xb_n(t_n)-\xb_n(t_n+T)| \\& < \ell .
\end{align*}
By arbitrariness of $\ell$, we thus conclude that $\xb(t+T)-\xb(t)=0$ for all $t\ge 0$. Hence, the proof of the proposition follows by arbitrariness of the solution   considered and by Theorem~\ref{thm.P_T_weak},  since  every steady-state trajectory   is $T$-periodic.
\ignorespaces\hfill$\QED$

\subsection{Proof of Proposition \ref{prop.xi}}\label{apd.proof.prop.xi} 
	As $g$ is $\cC^1$ there exists a continuous map $\vhi:\Rplus\x\R^m\to\R^m$  such that, for every $\chi_0\in\R^m$, the unique solution to~\eqref{s:xi} originating at $\chi_0$ satisfies $\chi(\cdot)=\vhi(\cdot,\chi_0)$. Then, by assumption, there exists $\Xi_\epsilon\subset\R^m$ and, for each $\signal\in\epsilon \ball(\cT_N)$,    there exists    $\chi_0\in\Xi_\epsilon$ such that $\signal(t)=\theta(\vhi(t,\chi_0))$ for all $t\ge 0$. Then,  with $(\alpha,\beta^1,\gamma^1,\dots,\beta^N,\gamma^N)=\phi_N(\signal)$,   we can write
	\begin{align*} 
	\int_0^{2\pi}\theta(\vhi(t,\chi_0)) dt &=  2\pi \alpha  \\
	\int_0^{2\pi}\theta(\vhi(t,\chi_0)) \sin(nt) dt &=  \pi \beta^n & n&=1,\dots,N\\
	\int_0^{2\pi}\theta(\vhi(t,\chi_0)) \cos(nt) dt &=  \pi \gamma^n & n&=1,\dots,N.
	\end{align*}
	Since this holds for all $\signal\in\epsilon\ball(\cT_N)$, then the above equations define a surjective mapping $\ell :\Xi_\epsilon \to\epsilon\ball(\R^{2N+1})$.
	  
\MB{As both $g$ and $\theta$ are $\cC^1$, so is $\ell$. Thus, $\ell$ is a surjective $\cC^1$ map from $\Xi_\epsilon$ to a subset of positive measure of $\R^{2N+1}$. 
Now, suppose that $m<2N+1$. Then, $1>0=\max\{0,m-(2N+1)\}$ and the rank  of the Jacobian of $\ell$ is always strictly less than $2N+1$ everywhere in $\Xi_\epsilon$. Hence, the application of the Morse-Sard  Theorem (see, e.g. \cite[Chapter~3]{Hirsch1994}) implies that $\ell(\Xi_\epsilon)$ has measure zero in $R^{2N+1}$, which contradicts the fact that $\ell(\Xi_\epsilon)=\epsilon\ball(\R^{2N+1})$. Hence, $m\ge 2N+1$.} 
	\ignorespaces\hfill$\QED$

\subsection{Proof of Lemma \ref{lem.sin}}\label{apd.proof.lem.sin}
	Suppose that for some $n\in\N$ two functions $r_n,\,q_n:\R^2\to\R$ are given that satisfy
\begin{equation}\label{rnqn}
\begin{aligned}
r_n(\sin(t), \cos(t)) &= \sin(nt), &
q_n (\sin(t),\cos(t)) &= \cos(nt)
\end{aligned}
\end{equation}
then the functions $r_{n+1}(a,b)  := r_n(a,b)b+q_n(a,b)a$ and $q_{n+1}(a,b)  := q_n(a,b)b-r_n(a,b)a$
satisfy
\begin{align*}
r_{n+1}\big(\sin(t), \cos(t) \big) &= \sin((n+1)t)\\
q_{n+1} \big(\sin(t), \cos(t) \big) &= \cos((n+1)t).
\end{align*}
Since for $n=1$ \eqref{rnqn} are trivially satisfied by the projections $r_1(a,b)=a$ and $q_1(a,b)=b$, then we claim by induction that for each $n\in\N$ there exist functions $r_n,q_n:\R^2\to\R$ satisfying \eqref{rnqn}. 	Fix $\epsilon>0$ and $N\in\N$ and let 
\begin{equation*}
\delta \!=\! \dfrac{\epsilon}{\left( 1\!+\!\sum_{n=1}^N\left( \displaystyle\sup_{(a,b)\in K} |r_n(a,b)|+\displaystyle\sup_{(a,b)\in K} |q_n(a,b)|  \right)\right)}.
\end{equation*}
Pick $\signal\in\delta\ball(\cT_N)$, then it satisfies
\begin{align*}
\signal(t)&=\alpha  + \sum_{n=1}^N \big( \beta^n \sin(nt) +\gamma^n\cos(nt)\big)\\
&=\alpha + \sum_{n=1}^N \big( \beta^n r_n(\sin(t),\cos(t)) +\gamma^n q_n(\sin(t),\MB{\cos(t)})\big)
\end{align*}
for some \MB{$|(\alpha,\beta^1,\gamma^1,\dots,\beta^N,\gamma^N)|<\delta$}. Hence, the function
\begin{equation*}
c_\signal(a,b):= \alpha  + \sum_{n=1}^N \big( \beta^n  r_n(a,b) +\gamma^nq_n(a,b)\big)
\end{equation*}
satisfies $c_\signal(\sin(t),\cos(t))=\signal(t)$ and, moreover, it fulfills
\begin{align*}
&\sup_{(a,b)\in K} |c_\signal(a,b)|\\ &\ \le \left( 1+\sum_{n=1}^N\left( \sup_{(a,b)\in K} |r_n(a,b)|+\sup_{(a,b)\in K} |q_n(a,b)|  \right)\right)|a| \\&\ < \epsilon . 
\end{align*} 
\null\ignorespaces\hfill$\QED$ 
 
\bibliographystyle{ieeetran}
\bibliography{biblio}

\end{document}